\documentclass[a4paper,fleqn]{cas-sc}

\usepackage[numbers]{natbib}
\usepackage{xcolor}
\usepackage{hyperref}
\usepackage{subcaption}
\usepackage{tikz}
\usetikzlibrary{arrows.meta,decorations.markings}
\usepackage{placeins}

\DeclareMathOperator{\sech}{sech}

\newtheorem{theorem}{Theorem}
\newtheorem{lemma}{Lemma}
\newdefinition{remark}{Remark}
\newproof{proof}{Proof}
\newproof{pot}{Proof of Theorem \ref{thm}}
\graphicspath{{Figs/}}

\definecolor{myred}{rgb}{0.7,0.1,0.16}
\definecolor{myblue}{rgb}{0,0.32,0.7}
\definecolor{mygreen}{rgb}{0.133,0.545,0.133}

\begin{document}
\let\WriteBookmarks\relax
\def\floatpagepagefraction{1}
\def\textpagefraction{.001}

\title [mode = title]{Nonholonomic directional pursuit and evasion: Global feedbacks}
\shorttitle{Nonholonomic directional pursuit and evasion: Global feedbacks}
\shortauthors{B. Wang and M. Krsti\'c}

\author[1]{Bo Wang}[orcid=0000-0001-6047-1400]\cormark[1]
\ead{bwang1@ccny.cuny.edu}
\credit{}
\affiliation[1]{organization={Department of Mechanical Engineering, The City College of New York, The City University of New York},
            city={New York},
            state={NY 10031},
            country={USA}}

\author[2]{Miroslav Krsti\'c}[orcid=0000-0002-5523-941X]
\ead{mkrstic@ucsd.edu}
\credit{}
\affiliation[2]{organization={Department of Mechanical and Aerospace Engineering, University of California San Diego},
            city={La Jolla},
            state={CA 92093},
            country={USA}}

\cortext[1]{Corresponding author.}

\nonumnote{\textit{E-mail addresses:} {bwang1@ccny.cuny.edu} (B. Wang).}

\begin{abstract}
In a recent paper by the second coauthor, directional pursuit-evasion for strictly forward-moving nonholonomic vehicles was solved ``half-globally", namely under favorable initial line-of-sight conditions. In this paper, we develop feedback designs that achieve the directional pursuit and evasion objectives from arbitrary initial relative configurations.
We achieve globality with completely different approaches to both the design and the analysis.
Our designs are less aggressive in both the forward-speed and steering laws, allowing transient overshoot of the pursuer-evader range during global reorientation.
The only price that we pay for globality is that our feedback laws require a priori knowledge of the opponent's maximal turning rate, whereas in the half-global work, no known bound of the opponent's turning rate was assumed. For the pursuit problem, the feedback law guarantees finite-time capture with prescribed directional alignment. For the evasion problem, the feedback law guarantees capture avoidance with a prescribed safety margin and achieves spinaway under a decay condition on the pursuer's turning rate. We illustrate the global capture and spinaway with simulations. 
The analysis is based on an integral input-to-state stability type mechanism induced by an endogenous time dilation, together with finite-time coextinction and safety-margin persistence lemmas for singularly coupled scalar inequalities. 
\end{abstract}


\maketitle

\section{Introduction}\label{sec:introduction}

In \cite{krstic2026directional}, directional pursuit and directional evasion for forward-moving nonholonomic vehicles were formulated and solved under unlimited maneuvering capability of the opponent, but under a ``half-global'' restriction to the initial conditions. By half-global we mean that the initial line-of-sight (LOS) angle is limited to $(-90^\circ, 90^\circ)$, so that the vehicle is not initially oriented in an unfavorable direction relative to its objective.
The present paper removes this LOS restriction and establishes global feedback designs, at the price of assuming a known uniform bound on the opponent's maneuvering capability. The proposed controllers differ from those in \cite{krstic2026directional} in both design and analysis. In particular, the forward-speed and steering feedback laws are less aggressive, allowing transient overshoot of the range during reorientation. This overshoot is the mechanism that enables recovery from arbitrary initial LOS angles and leads to global directional pursuit and evasion guarantees.

\paragraph{Related work.}
Pursuit-evasion has been studied from several complementary viewpoints. Classical formulations are differential-game-theoretic, where capture and escape are characterized geometrically through game values, optimal strategies, and capture barriers \cite{isaacs1965differential,bacsar1998dynamic,exarchos2015suicidal,merz1972game,buzikov2023game,bera2017comprehensive}; computational approaches use Hamilton--Jacobi equations and reachability analysis \cite{mitchell2005time,sun2017multiple}. In mobile robotics, pursuit-evasion has also been studied through search, visibility maintenance, and motion planning \cite{chung2011search,guibas1999visibility,gerkey2006visibility,lozano2022visibility,ruiz2013time,ruiz2016capturing,bravo2020pursuit}. More closely related to the present work are feedback designs for nonholonomic pursuit-evasion, including proportional bearing-only pursuit with capture-or-tracking guarantees but not finite-time capture \cite{dovrat2021tracking}, two-phase worst-case strategies supported by simulations \cite{nath2022two}, nonlinear model predictive control formulations for obstacle-avoiding pursuit-evasion \cite{sani2020pursuit}, and nature-inspired Alert-Turn strategies with explicit capture-time estimates under restricted initial geometries \cite{zhou2026nature}. These results do not provide closed-form feedback laws that achieve prescribed directional pursuit and evasion objectives from arbitrary initial relative configurations. The directional viewpoint adopted here is also distinct from proportional-navigation guidance \cite{kreindler1973optimality,ryoo2004capturability,palumbo2010modern}, which is typically formulated for missile or interceptor models with lateral-acceleration or heading-rate commands rather than the nonholonomic forward-speed and steering-rate inputs considered in this paper.

\paragraph{The work that we advance.}
The recent formulation of directional pursuit and evasion for nonholonomic vehicles in \cite{krstic2026directional} introduced prescribed-angle interception and escape objectives and developed analytically explicit input-to-state-stabilizing (ISS) and inverse-optimal feedback designs \cite{krstic1998inverse}, without requiring the solution of Hamilton--Jacobi--Isaacs equations. This work provides a unified feedback framework for directional pursuit-evasion of nonholonomic vehicles and establishes rigorous ISS and optimality properties. Its guarantees, however, rely on favorable (half-global) initial heading conditions: capture is ensured when the pursuer does not initially head away from the evader, while escape is ensured when the evader initially heads away from the pursuer. Thus, despite this substantial progress, fully \textit{global} feedback design for directional pursuit and evasion of nonholonomic vehicles remains open. The present paper follows the same directional viewpoint, but develops feedback laws that guarantee the directional pursuit and evasion objectives from arbitrary initial relative configurations.

\paragraph{Contributions.}
The present work develops a robust nonlinear feedback framework for directional pursuit and evasion of forward-moving nonholonomic vehicles under kinematic advantage assumptions and unknown but bounded opponent steering inputs. Two complementary problems are addressed: \emph{directional pursuit of a maneuvering evader} and \emph{evasion against a maneuvering pursuer}.
\begin{itemize}
\item The first contribution is the construction of explicit closed-form feedback laws that are globally valid for the relative nonholonomic dynamics. For the pursuit problem, the proposed controller removes the favorable initial-heading restriction in \cite{krstic2026directional} and guarantees finite-time capture with the prescribed directional objective from arbitrary initial relative configurations. For the evasion problem, the proposed controller guarantees capture avoidance for any prescribed safety margin and achieves the desired spinaway configuration under an additional decay condition on the pursuer's steering input.

\item The second contribution is an integral input-to-state stability (iISS) formulation of the pursuer-evader interaction. In both pursuit and evasion, the opponent's unknown steering input is treated as an external input, and the closed-loop relative dynamics are analyzed through iISS-type Lyapunov estimates. This provides a robust nonlinear stability mechanism for handling bounded adversarial maneuvering without measuring or canceling the opponent's steering input.

\item The third contribution is a quantitative characterization of the resulting closed-loop behavior. In the pursuit problem, explicit estimates are obtained for the capture time, transient range excursion, angular-error excursion, and pursuer speed. In the evasion problem, the analysis gives an explicit positive lower bound on the distance from the safety boundary, together with ISS/iISS-type estimates for the angular subsystems. These estimates quantify the transient cost of achieving globality and clarify how the known steering bounds enter the gain selection.

\item The fourth contribution is to the foundations of stability analysis and of independent interest beyond pursuit-evasion. The finite-time capture mechanism is distilled into a coextinction lemma for a singularly coupled pair of scalar Lyapunov differential inequalities, in which the vanishing range variable and the angular reorientation error are driven to zero at the same terminal time. This result is complemented by an endogenous time-dilation lemma, which converts a finite-time singular cascade into an infinite-horizon iISS problem with an integrable disturbance budget. For the evasion problem, a safety-margin persistence lemma gives the corresponding mechanism for maintaining a positive separation from the capture boundary. These auxiliary results are stated abstractly and may be useful for other singularly coupled nonlinear control problems.
\end{itemize}

\paragraph{Relation to prior work of the authors.}
The authors' prior work on polar-coordinate feedback for nonholonomic systems concerns single-agent stabilization and formation: strict Lyapunov functions for unicycle stabilization \cite{han2024safety,wang2026further,todorovski2025modular,todorovski2026nonholonomicrobotparkingfeedback}, explicit polar-coordinate feedback laws \cite{wang2021formation,krstic2025integrator}, and inverse-optimal designs \cite{kim2025inverse,kim2025nonholonomicrobotparkingfeedback}. The present paper differs in studying a two-agent pursuit-evasion interaction with competing objectives, where the opponent's steering input enters as an external disturbance.

Compared with \cite{krstic2026directional}, the present paper removes the half-global initial-heading restrictions and obtains global directional pursuit and evasion guarantees. This global extension relies on a known uniform bound on the opponent's steering input, which is used explicitly in the gain selection. Accordingly, the robustness mechanism changes from ISS in \cite{krstic2026directional} to an iISS-type analysis in the present work. Besides a completely new technical development in feedback design, another enabling idea we present  pertains to the endogenous  reparameterization of time;  unlike the reparametrization introduced in \cite{krstic2026directional}, which requires the range to be monotone, the range-reparametrized time we introduce here allows a non-monotone-in-time range.

\section{Problem statement}\label{sec:problem}

Consider a pursuer-evader pair of forward-moving nonholonomic vehicles moving in the plane. 
Each vehicle is modeled by the standard unicycle kinematics
\begin{equation}\label{eq:unicycle}
    \dot{x}_i = v_i\cos\theta_i,\qquad \dot{y}_i = v_i\sin\theta_i,\qquad \dot{\theta}_i = \omega_i,
\end{equation}
where $(x_i,y_i)\in\mathbb{R}^2$ is the position, $\theta_i\in\mathbb{R}$ is the heading angle,
$v_i(t)>0$ denotes the admissible forward speed, and $\omega_i(t)\in\mathbb{R}$ denotes the steering rate. The index $i\in\{\mathrm{p},\mathrm{e}\}$ identifies the pursuer and the evader. 

Let us introduce the pursuer-evader relative error variables in polar coordinates as
\begin{equation}
    \rho:=\sqrt{(x_{\mathrm{p}}-x_{\mathrm{e}})^2+(y_{\mathrm{p}}-y_{\mathrm{e}})^2},\quad \delta:=\operatorname{atan2}(y_{\mathrm{p}}-y_{\mathrm{e}},x_{\mathrm{p}}-x_{\mathrm{e}})+\pi-\theta_{\mathrm{e}},\quad 
    \gamma:=\delta+\theta_{\mathrm{e}}-\theta_{\mathrm{p}},
\end{equation}
where $\rho$ represents the relative distance, $\delta$ denotes the relative polar angle, and $\gamma$ denotes the relative line-of-sight (LOS) angle\footnote{Throughout the paper, angular variables are understood as continuous representatives on the universal covering space of the circle, obtained by unwrapping when necessary.}, as illustrated in Fig.~\ref{fig:MR}. 
For $\rho>0$, the corresponding pursuer-evader error dynamics are given by
\begin{subequations}\label{eq:3}
    \begin{eqnarray}
        \dot{\rho}&=&v_{\mathrm{e}}\cos\delta - v_{\mathrm{p}}\cos\gamma \label{eq:3a}\\
        \dot{\delta}&=&\frac{1}{\rho}(v_{\mathrm{p}} \sin\gamma - v_{\mathrm{e}}\sin\delta) - \omega_{\mathrm{e}} \label{eq:3b}\\
        \dot{\gamma}&=&\frac{1}{\rho}(v_{\mathrm{p}} \sin\gamma - v_{\mathrm{e}}\sin\delta) - \omega_{\mathrm{p}}. \label{eq:3c}
    \end{eqnarray}
\end{subequations}
Throughout the paper, for any trajectory starting at $t=0$, we use the subscript $0$ to denote the initial value of a state variable; in particular,
$\rho_0:=\rho(0)$, $\delta_0:=\delta(0)$, and $\gamma_0:=\gamma(0)$.

\begin{figure}
    \centering
    \includegraphics[width=0.35\linewidth]{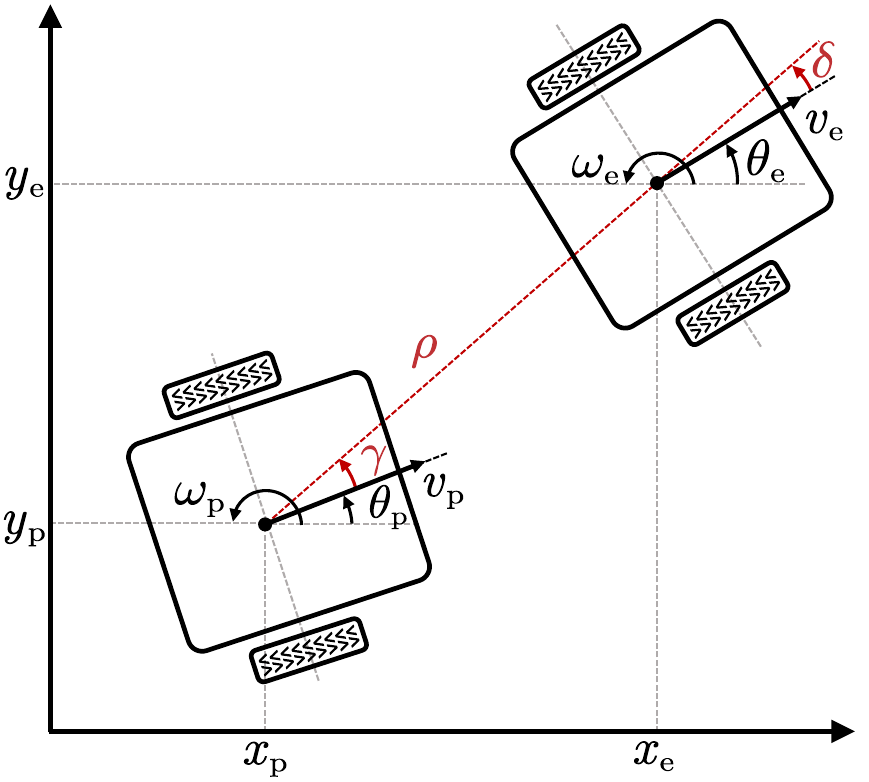}
    \caption{\unboldmath Pursuer configuration $(x_\mathrm{p},y_\mathrm{p},\theta_\mathrm{p})$ relative to the evader configuration $(x_\mathrm{e},y_\mathrm{e},\theta_\mathrm{e})$, and the corresponding relative position and angles in polar coordinates $(\rho,\delta,\gamma)$.}
    \label{fig:MR}
\end{figure}

\paragraph{Problem statement.} We consider the directional pursuit and directional evasion problems of \cite{krstic2026directional}, under the additional information assumption that the opponent's steering input is unknown but uniformly bounded, and with a prescribed safety margin $\rho_{\min}>0$ in the evasion case:
\begin{enumerate}
    \item[(P1)] \textit{Directional pursuit of a maneuvering evader.} Design a controller for a kinematically advantaged pursuer such that: (i) \textit{finite-time capture} is achieved, i.e., there exists a finite time \(T<\infty\) for which \(\rho(t)>0\) for all \(t<T\), and \(\lim_{t\to T^-}\rho(t)=0\); and (ii) \textit{directional alignment} is attained in the sense that the relative angular errors satisfy $\delta(t)\to 0$ and $\gamma(t)\to 0$ as $t\to T^-$.  

    \item[(P2)] \textit{Evasion against a maneuvering pursuer.} Design a controller for a kinematically advantaged evader such that: (i) \textit{capture avoidance} is guaranteed for any prescribed safety margin $\rho_{\min}<\rho_0$;    
    and (ii) \textit{asymptotic spinaway} is achieved, in the sense that the evader asymptotically moves behind the pursuer while facing away from it, namely, the relative angular errors satisfy $\delta(t)\to 0$ and $\gamma(t)\to \pi$ as $t\to \infty$.
\end{enumerate}

\section{Directional pursuit of a maneuvering evader}\label{sec:pursuer}

In this section, we develop a directional pursuit law, in the sense of \cite{krstic2026directional}, with a bounded smooth speed control that guarantees finite-time capture of bounded maneuvering evaders from arbitrary initial relative configurations.

\begin{theorem}\label{thm:2}
Consider \eqref{eq:3}. 
Assume that the evader moves with constant speed \(v_{\mathrm e}>0\) and has a piecewise continuous 
steering input satisfying 
\begin{equation}
    |\omega_{\mathrm e}(t)|\le \omega_M,\quad \forall t\ge0 .
\end{equation}
for a known $\omega_M>0$. Let the pursuer's linear velocity control be 
\begin{equation}\label{eq:vp}
v_{\mathrm{p}} := (1+\varepsilon)\, v_{\mathrm{e}} \sqrt{1+\tanh^2(c_0\delta)},
\end{equation}
and let the pursuer's steering control be
\begin{subequations}\label{eq:wp}
    \begin{eqnarray}
    \omega_{\mathrm{p}}&:=& (1+a(\delta))\frac{v_{\mathrm{p}}\sin\gamma - v_{\mathrm{e}} \sin \delta}{\rho} + \left(k+\frac{c}{\rho}\right)\tilde{\gamma} \label{eq:5a}\\
    \tilde{\gamma}&:=&\gamma-\gamma^*(\delta) \label{eq:5b}\\
    \gamma^*(\delta)&:=& -\arctan(\tanh(c_0\delta)),\label{eq:5c}\\
    a(\delta)&:=&\frac{c_0\sech^2(c_0\delta)}{1 + \tanh^2(c_0\delta)}\in(0,c_0], 
    \end{eqnarray}
\end{subequations}
where $\varepsilon>0$, $c_0^\star(\varepsilon):=\ln((2+\varepsilon)/\varepsilon)/(2\pi)$, $c_0\ge \max\left\{1,c_0^\star(\varepsilon)\right\}$, and $c>\varepsilon v_{\mathrm{e}}$. Let $\eta>0$ satisfy
\begin{equation}\label{eq:eta}
    (1+\varepsilon)\left(\eta+\frac{\eta^2}{2} \right)\le \frac{\varepsilon}{2}
\end{equation}
and choose $k>c_0\omega_M/\eta$.
Then, for every initial condition $\rho_0>0$, $\delta_0\in\mathbb{R}$, and $\gamma_0\in\mathbb{R}$, the closed-loop system admits a unique maximal solution on $[0,T)$, with $\rho(t)>0$ on $[0,T)$ and
\begin{equation}\label{eq:rho_t}
    \lim_{t\to T^-}(\rho(t),\delta(t),\gamma(t))=(0,0,0),
\end{equation}
where
\begin{equation}\label{eq:T1}
    T\le T_2:=\frac{2\rho_0}{\varepsilon v_{\mathrm e}}+\frac{6+5\varepsilon}{\varepsilon}T_1, \quad T_1:=\max\left\{0,\frac{1}{k}\ln\frac{|\gamma_0+\arctan(\tanh(c_0\delta_0))|}{\eta-c_0\omega_M/k}\right\},
\end{equation}
with the convention $\ln(0):=-\infty$.
Furthermore, the closed-loop solution remains bounded on its whole interval of existence, i.e., there exists $Q\in\mathcal{K}$ such that
\begin{equation}\label{eq-stab1}
|(\rho(t),\delta(t),\gamma(t))| \leq Q(|(\rho_0,\delta_0,\gamma_0)|), \quad \forall t\in[0,T)\,,
\end{equation}
and in particular
\begin{subequations}
    \begin{eqnarray}
        \max_{t\in[0,T)}\rho(t)&\le&\rho_0+\frac{(1+\varepsilon)v_e}{k}|\tilde\gamma_0|\left[1+\frac{c_0\omega_M}{k}+\frac{1}{4}|\tilde\gamma_0|\right],\label{eq:bound_rho}\\
        \max_{t\in[0,T)} |\tilde{\gamma}(t)|&\le&\max\left\{|\tilde\gamma_0|, \frac{c_0\omega_M}{k} \right\},\label{eq:bound_tilde_gamma}
    \end{eqnarray}
\end{subequations}
where $|\tilde\gamma_0|=|\gamma_0+\arctan(\tanh(c_0\delta_0))|$.
Finally, the speed is uniformly bounded as
\begin{equation}\label{eq:bound_vp}
    (1+\varepsilon)v_\mathrm{e}\le v_\mathrm{p}(t)< \sqrt{2}(1+\varepsilon)v_\mathrm{e}, \quad \forall t\in[0,T),
\end{equation}
and the steering input satisfies $\omega_\mathrm{p}\in L^\infty([0,T))$. 
\end{theorem}

Before proving Theorem~\ref{thm:2}, we provide several remarks to clarify the statement and explain the main mechanisms behind the global capture result.

\begin{remark}[Comparison with the half-global pursuit law in \cite{krstic2026directional}]\rm

It is useful to compare the pursuit law in Theorem~\ref{thm:2} with the half-global pursuit law in \cite[Theorem 2]{krstic2026directional}, written with notation adapted to the present paper. In \cite{krstic2026directional}, the pursuer speed and steering have the structure
\begin{equation}
v_{\mathrm{p}}=\frac{(1+\varepsilon)v_{\mathrm{e}}}{\cos\gamma},
\quad
\omega_{\mathrm{p}}=\frac{1}{\rho}\left(c\sec^2\gamma\right)\tilde\gamma,
\quad
\tilde\gamma=\tan\gamma+c_0\delta,
\end{equation}
on the favorable domain where $\cos\gamma>0$. 
Thus, the two designs share the terminal $1/\rho$ structure that drives capture, but differ essentially in how they handle unfavorable line-of-sight angles. The controller in \cite[Theorem 2]{krstic2026directional} uses the velocity law $v_{\mathrm{p}}=(1+\varepsilon)v_e/\cos\gamma$, which becomes singular at the orthogonal LOS angle $\gamma=\pm\pi/2$. The same singularity appears in the backstepping coordinate through $\tan\gamma$ and in the steering gain through $\sec^2\gamma$. These singular factors prevent the controller from crossing the orthogonal line of sight and lead to a half-global result.

Here, $v_{\mathrm p}\in[(1+\varepsilon)v_{\mathrm e},\sqrt{2}(1+\varepsilon)v_{\mathrm e})$ and $\gamma^*(\delta)\in(-\pi/4,\pi/4)$ are bounded, and $\tilde\gamma=\gamma+\arctan(\tanh(c_0\delta))$ and $a(\delta)$ are globally regular, so no singularity occurs at $\gamma=\pm\pi/2$.
\end{remark}

\begin{remark}[Range excursion and the cost of global reorientation]\rm
The bound \eqref{eq:bound_rho} quantifies the initial range excursion caused by the reorientation error $|\tilde\gamma_0|$: it vanishes as $|\tilde\gamma_0|\to 0$ and decreases as $k$ increases. The nonmonotonicity of $\rho(t)$ has a clear physical meaning: from an unfavorable initial heading, the pursuer may initially move farther away while turning toward the evader; once the reorientation error is sufficiently reduced, the range decreases uniformly, and capture follows.
\end{remark}

\begin{remark}[Projection design and the role of saturation]\rm
The control law \eqref{eq:vp}-\eqref{eq:wp} is obtained from a projection design: $\gamma^*(\delta)$ and $v_{\mathrm p}$ are chosen jointly so that
\begin{subequations}
    \begin{eqnarray}
        v_{\mathrm{p}}\cos\gamma^*&=&(1+\varepsilon)v_{\mathrm{e}},\\
        v_{\mathrm{p}}\sin\gamma^*&=&-(1+\varepsilon)v_{\mathrm{e}}\tanh(c_0\delta);
    \end{eqnarray}
\end{subequations}
solving gives \eqref{eq:vp} and \eqref{eq:5c}. The first projection makes the range strictly decreasing on $\tilde\gamma=0$; the second introduces the stabilizing term in $\dot\delta$. The saturation $\tanh(c_0\delta)$ is what keeps $v_{\mathrm p}$ and $\gamma^*$ bounded and renders $K_1,K_2$ globally bounded and uniformly vanishing as $\tilde\gamma\to 0$.
\end{remark}

\begin{remark}[Gain selection]\rm 
The free saturation gain $c_0$ removes the artificial lower bound on the kinematic advantage $\varepsilon$: $c_0\ge c_0^\star(\varepsilon)$ ensures sign-definiteness of $g$ (Lemma \ref{lem:1}), while $c_0\ge1$ simplifies the terminal damping estimates used to prove boundedness of the steering input. Hence $\varepsilon>0$ is required only to provide the strict closing margin in the range dynamics. The conditions \eqref{eq:eta} and $k>c_0\omega_M/\eta$ have a simple conservative interpretation. Choosing $\eta(\varepsilon):=\varepsilon/(4+4\varepsilon)$ verifies \eqref{eq:eta}, since $(1+\varepsilon)(\eta+\eta^2/2)\le 9\varepsilon/32<\varepsilon/2$, and then $k>c_0\omega_M/\eta$ becomes $k>4c_0(1+\varepsilon)\omega_M/\varepsilon$. Thus, for small $\varepsilon$, the required gain scales as $k=\mathcal{O}(c_0\omega_M/\varepsilon)$: stronger evader maneuvering and smaller kinematic advantage both require larger $k$.
\end{remark}

\begin{proof}[Theorem \ref{thm:2}]~

\paragraph{1. Closed-loop system and well-posedness.} 

For later reference, define
\begin{subequations}\label{eq:func}
    \begin{eqnarray}
    f_\rho(\delta)&:=&-v_{\mathrm e}(1+\varepsilon-\cos\delta),\label{eq:f_rho}\\
    f_\delta(\delta)&:=&-v_{\mathrm e}\left((1+\varepsilon)\tanh(c_0\delta)+\sin\delta\right), \label{eq:f_delta}\\
    K_1(\delta,\tilde\gamma)&:=&(1+\varepsilon)v_{\mathrm e}
\left[(1-\cos\tilde\gamma)-\tanh(c_0\delta)\sin\tilde\gamma\right],\label{eq:K1}\\
    K_2(\delta,\tilde\gamma)&:=&(1+\varepsilon)v_{\mathrm e}
\left[\tanh(c_0\delta)(1-\cos\tilde\gamma)+\sin\tilde\gamma\right],\label{eq:K2}
    \end{eqnarray}
\end{subequations}
Substituting \eqref{eq:vp}, \eqref{eq:5b}, and \eqref{eq:5c} into \eqref{eq:3} and applying the identities
\begin{equation}\label{eq:cossin}
    \cos(-\arctan(r)+s)=\frac{\cos(s)+r\sin(s)}{\sqrt{1+r^2}},\quad \sin(-\arctan(r)+s)=\frac{\sin(s)-r\cos(s)}{\sqrt{1+r^2}},
\end{equation}
followed by the trivial add-and-subtract of $(1+\varepsilon)v_{\mathrm e}$ in $\dot\rho$ and of $(1+\varepsilon)v_{\mathrm e}\tanh(c_0\delta)$ in $\dot\delta$, yields the closed-loop system
\begin{subequations}\label{eq:14}
	\begin{eqnarray}
		\dot{\rho}
		&=&
		f_\rho(\delta)+K_1(\delta,\tilde{\gamma})
		\label{eq:26a}\\
        \dot{\delta}
        &=&
        \frac{1}{\rho}
        [f_\delta(\delta) + K_2(\delta,\tilde{\gamma})]
        -\omega_{\mathrm{e}}
        \label{eq:26b}\\
        \dot{\tilde{\gamma}}
        &=&
        -\left(k+\frac{c}{\rho}\right)\tilde{\gamma} - a(\delta)\omega_{\mathrm{e}},
        \label{eq:26c}
	\end{eqnarray}
\end{subequations}
where \eqref{eq:26c} uses $\dot{\tilde\gamma}=\dot\gamma-\dot\gamma^*=\dot\gamma+a(\delta)\dot\delta$ from \eqref{eq:5b}-\eqref{eq:5c} and the steering law \eqref{eq:5a}. The functions $f_\rho,f_\delta,K_1,K_2$ are globally bounded, with $f_\rho\in[-(2+\varepsilon)v_{\mathrm e},-\varepsilon v_{\mathrm e}]$ and $K_1\in(-(1+\varepsilon)v_{\mathrm e},3(1+\varepsilon)v_{\mathrm e})$, with $\delta f_\delta(\delta)<0$ for $\delta\ne 0$ by Lemma~\ref{lem:1}, and $K_1,K_2\to 0$ as $\tilde\gamma\to 0$ uniformly in $\delta$. Since the right-hand side of \eqref{eq:14} is locally Lipschitz on $(0,\infty)\times\mathbb R^2$ and bounded whenever $\rho$ is bounded away from zero, the closed-loop system admits a unique maximal solution on $[0,T)$ with $\rho(t)>0$ on $[0,T)$; if $T<\infty$, then $\rho(t)\to 0$ as $t\to T^-$.

\paragraph{2. Finite-time capture.} 
The capture properties of the theorem are obtained by verifying the hypotheses of Lemma~\ref{lem:capture} (Appendix~\ref{app:capture}) for the pair $(\rho,|\tilde\gamma|)$ governed by \eqref{eq:26a} and \eqref{eq:26c}, and reading off its conclusions.

\par\medskip\noindent
\textbf{Verification of the hypotheses.} From \eqref{eq:26c}, using $a(\delta)\in(0,c_0]$ and $|\omega_{\mathrm{e}}(t)|\le\omega_M$, the reorientation error obeys
\begin{equation}\label{eq:hyp_w}
  \frac{\rm d}{{\rm d}t}|\tilde{\gamma}|
  \le
  -\left(k+\frac{c}{\rho}\right)|\tilde{\gamma}| + c_0\omega_M .
\end{equation}
Using the inequalities $|1-\cos(z)|\le z^2/2$, $|\sin(z)|\le|z|$, and $|\tanh(c_0\delta)|\le1$, one has
\begin{equation}\label{eq:K1_bound}
    |K_1(\delta,\tilde{\gamma})|\le (1+ \varepsilon) v_{\mathrm{e}} \left( \frac{\tilde{\gamma}^2}{2}+|\tilde{\gamma}|  \right),
\end{equation}
so that, by \eqref{eq:26a} and the bounds on $f_\rho, K_1$,
\begin{equation}\label{eq:hyp_rho}
  -(3+2\varepsilon)v_{\mathrm{e}}
  \le
  \dot{\rho}=f_\rho(\delta)+K_1(\delta,\tilde{\gamma})
  \le
  -\varepsilon v_{\mathrm{e}}+(1+\varepsilon)v_{\mathrm{e}}\left(|\tilde{\gamma}|+\frac{\tilde{\gamma}^2}{2}\right).
\end{equation}
Thus, the hypotheses \eqref{eq:lem-w}-\eqref{eq:lem-rho} of Lemma~\ref{lem:capture} hold with
\begin{equation}\label{eq:dictionary}
  w=|\tilde{\gamma}|,\quad b=c_0\omega_M,\quad q=\varepsilon v_{\mathrm{e}},\quad L=L_\rho:=(3+2\varepsilon)v_{\mathrm{e}},\quad \chi(w)=(1+\varepsilon)v_{\mathrm{e}}\left(w+\frac{w^2}{2}\right),
\end{equation}
and the design gains $c,k$. The inner floor is $\underline w=b/k=c_0\omega_M/k$. Moreover, condition \eqref{eq:eta} gives $\chi(\eta)\le q/2$, and hence $\eta\le w^\star:=\chi^{-1}(q/2)$. Therefore the gain condition $k>c_0\omega_M/\eta$ implies $b/k<\eta\le w^\star$, which is precisely the requirement in Lemma~\ref{lem:capture}.

\par\medskip\noindent
\textbf{Conclusions of Lemma \ref{lem:capture}.}
With $\alpha:=c/L_\rho$, Lemma~\ref{lem:capture} yields
\begin{itemize}
\item[(i)] the ultimate bound \eqref{eq:bound_tilde_gamma}; Explicitly, the Step-1 majorant is
\begin{equation}\label{eq:tilde_gamma_bound}
    |\tilde{\gamma}(t)|\le e^{-kt}|\tilde\gamma_0| + \frac{c_0\omega_M}{k}(1 - e^{-kt});
\end{equation}
\item[(ii)] finite-time capture \eqref{eq:rho_t}; Since $K_1<3(1+\varepsilon)v_{\mathrm{e}}$, the excursion admits the crude bound $\Omega\le(3+2\varepsilon)v_{\mathrm{e}}T_1$; together with the entry time $T_1$ majorized as in \eqref{eq:T1}, the lemma's bound $T_1+{2}(\rho_0+\Omega(w_0))/q$ gives the explicit capture-time estimate $T\le T_2$ of \eqref{eq:T1}, with underlying two-phase majorant
\begin{equation}\label{eq:rho_bar}
\boxed{
    \rho(t)\le \bar{\rho}(t):=
    \begin{cases}
\rho_0 + (3+2\varepsilon)v_{\mathrm{e}} t, 
& 0\le t< T_1,\\[2mm]
 \rho_0+(3+2\varepsilon)v_{\mathrm{e}} T_1 
-\dfrac{\varepsilon v_{\mathrm{e}}}{2}(t-T_1),
& T_1\le t < T_2;
\end{cases}}
\end{equation}
\item[(iv)] terminal convergence $\tilde{\gamma}(t)\to0$ as $t\to T^-$, whose explicit terminal estimate \eqref{eq:44} is recorded in Part~3.
\end{itemize}

The excursion bound (iii) of Lemma~\ref{lem:capture} can be sharpened into the initial-condition-explicit form \eqref{eq:bound_rho} for the present $\chi$. Since $c_0\omega_M/k<\eta$, condition \eqref{eq:eta} implies
\begin{equation}
(1+\varepsilon)
\left(
\frac{c_0\omega_M}{k}
+
\frac{1}{2}
\left(\frac{c_0\omega_M}{k}\right)^2
\right)
\le
\frac{\varepsilon}{2}.
\end{equation}
Using \eqref{eq:26a}, \eqref{eq:tilde_gamma_bound}, and \eqref{eq:K1_bound}, one obtains 
\begin{subequations}
    \begin{eqnarray}
        \dot\rho(t)
&\le&-\varepsilon v_e+(1+\varepsilon)v_e\left(|\tilde\gamma(t)|+\frac{1}{2}|\tilde\gamma(t)|^2
\right)\\
&\le&-\frac{\varepsilon v_e}{2}+(1+\varepsilon)v_e\left[\left(1+\frac{c_0\omega_M}{k}\right)
e^{-kt}|\tilde\gamma_0|+\frac{1}{2}e^{-2kt}|\tilde\gamma_0|^2\right].
    \end{eqnarray}
\end{subequations}
Integrating this estimate gives, before capture,
\begin{equation}
\rho(t)\le\rho_0+\frac{(1+\varepsilon)v_e}{k}|\tilde\gamma_0|\left[1+\frac{c_0\omega_M}{k}+\frac{1}{4}|\tilde\gamma_0|\right]-\frac{\varepsilon v_e}{2}t,
\end{equation}
which yields \eqref{eq:bound_rho}.

\paragraph{3. Angular convergence at capture.} From the closed-loop \eqref{eq:26a} and the bounds on $f_\rho, K_1$, $|\dot\rho(t)|\le L_\rho:=(3+2\varepsilon)v_{\mathrm e}$ for all $t\in[0,T)$. Since $\rho(T)=0$, for any $t\in[0,T)$,
\begin{equation}\label{eq:40}
    \rho(t)=\left|\int_t^T\dot{\rho}(\tau) {\rm d}\tau\right| \le L_\rho (T-t)
    \quad\Rightarrow\quad
    \frac{1}{\rho(t)}\ge \frac{1}{L_\rho(T-t)}.
\end{equation}
The terminal mechanism of Lemma~\ref{lem:capture} (Step~3 of its proof) specializes here through the explicit gain bound \eqref{eq:40}. With $\alpha:=c/L_\rho>0$, substituting $c/\rho(t)\ge\alpha/(T-t)$ into \eqref{eq:hyp_w}, applying the comparison principle, and dropping the factor $e^{-k(t-\tau)}\le1$ gives the explicit terminal estimate
\begin{equation}\label{eq:44}
\boxed{
|\tilde{\gamma}(t)|
\le
e^{-kt}\left(\frac{T-t}{T}\right)^{\alpha} |\tilde\gamma_0|
+
c_0\omega_M \int_0^t
\left(\frac{T-t}{T-\tau}\right)^{\alpha}  {\rm d}\tau.}
\end{equation}
By Lemma~\ref{lem:capture}(iv) the right-hand side tends to zero, so $|\tilde{\gamma}(t)|\to0$ as $t\to T^-$; the estimate \eqref{eq:44} is of the ISS kind, with a vanishing input effect at $t=T$.

Next, we analyze the bearing dynamics \eqref{eq:26b}. Written out, $\dot\delta=\frac{1}{\rho}\left(f_\delta(\delta)+K_2(\delta,\tilde\gamma)-\rho\,\omega_{\mathrm e}\right)$ is a finite-time singular cascade: the bearing is driven at a rate $1/\rho$ that diverges as the capture variable $\rho\to0$. Its convergence is supplied by Lemma~\ref{lem:dilation} (Appendix~\ref{app:dilation}), whose hypotheses we now verify and whose conclusion we then read off.

\par\medskip\noindent
\textbf{Endogenous clock.} By Part~2, $\rho(t)\to0$ at the finite time $T$, and \eqref{eq:40} gives $\rho(t)\le L_\rho(T-t)$ with $L_\rho=(3+2\varepsilon)v_{\mathrm e}$; thus $\rho$ is an endogenous finite-time clock. Lemma~\ref{lem:dilation} introduces the stretched time $s(t):=\int_0^t\rho(\tau)^{-1}{\rm d}\tau$, a strictly increasing bijection of $[0,T)$ onto $[0,\infty)$, in which \eqref{eq:26b} desingularizes to
\begin{equation}\label{eq:51}
    \frac{{\rm d}\delta}{{\rm d}s}=f_\delta(\delta)+d(s),\qquad s\in[0,\infty),
\end{equation}
with additive input
\begin{equation}\label{eq:d}
    d(s):=K_2\left(\delta(t(s)),\tilde\gamma(t(s))\right)-\rho(t(s))\,\omega_{\mathrm e}(t(s)).
\end{equation}

\par\medskip\noindent
\textbf{Intrinsic iISS.} The bearing equation \eqref{eq:51} admits an iISS Lyapunov function with linearly bounded supply rate. Indeed, with
\begin{equation*}
    V_\delta(\delta):=-\int_0^\delta f_\delta(l)\,{\rm d}l
    = v_{\mathrm e}\left[\frac{1+\varepsilon}{c_0}\ln(\cosh(c_0\delta))+1-\cos\delta\right],
\end{equation*}
positive definite by Lemma~\ref{lem:1}, the inequalities $\ln(\cosh(r))\le r$, $1-\cos(r)\le r^2/2$, and $c_0\ge1$ give properness,
\begin{equation}\label{eq:proper}
    \alpha_{1\delta}(|\delta|)\le V_\delta(\delta)\le\alpha_{2\delta}(|\delta|),\quad
    \alpha_{1\delta}(r):=\frac{(1+\varepsilon)v_{\mathrm e}}{c_0}\ln(\cosh(c_0 r)),\ \
    \alpha_{2\delta}(r):=v_{\mathrm e}\left[(1+\varepsilon)r+\frac{r^2}{2}\right],
\end{equation}
with $\alpha_{1\delta},\alpha_{2\delta}\in\mathcal K_\infty$. Along \eqref{eq:51},
\begin{equation}\label{eq:53b}
    \frac{{\rm d}V_\delta}{{\rm d}s}=-f_\delta^2(\delta)-f_\delta(\delta)\,d(s)
    \le -f_\delta^2(\delta)+(2+\varepsilon)v_{\mathrm e}|d(s)|,
\end{equation}
with $f_\delta^2$ positive definite (Lemma~\ref{lem:1}) and $(2+\varepsilon)v_{\mathrm e}|\cdot|\in\mathcal K_\infty$; hence $V_\delta$ is an iISS Lyapunov function for \eqref{eq:51}.

\par\medskip\noindent
\textbf{$\rho$-modulated disturbance.} In \eqref{eq:d}, the term $\rho\,\omega_{\mathrm e}$ is of the form $\rho\times(\text{bounded})$, since $|\omega_{\mathrm e}|\le\omega_M$. For the remaining term,
\begin{equation}\label{eq:K2_bound}
    |K_2(\delta,\tilde\gamma)|\le (1+\varepsilon)v_{\mathrm e}\left(|\tilde\gamma|+\frac12\tilde\gamma^2\right),
\end{equation}
and in stretched time, the reorientation error is the output of the exponentially stable filter
\begin{equation}\label{eq:gamma_s_dyn}
    \frac{{\rm d}\tilde\gamma}{{\rm d}s}=-(c+k\rho)\,\tilde\gamma-\rho\,a(\delta)\,\omega_{\mathrm e},
    \qquad |a(\delta)|\le c_0,
\end{equation}
of decay rate at least $c>0$ and forcing $-\rho\,a(\delta)\omega_{\mathrm e}$ again of the form $\rho\times(\text{bounded})$. Both channels of $d$ are therefore of the $\rho$-modulated type admitted by Lemma~\ref{lem:dilation}(ii).

\par\medskip\noindent
\textbf{Conclusion of Lemma~\ref{lem:dilation}.} 
The hypotheses of Lemma~\ref{lem:dilation} hold. Its budget identity $\int_0^\infty \rho(s) {\rm d}s=T<\infty$ implies $\rho\omega_e\in L^1([0,\infty))$ and $\rho a(\delta)\omega_e\in L^1([0,\infty))$. Since $\tilde\gamma$ satisfies the exponentially stable filter \eqref{eq:gamma_s_dyn} with decay rate at least $c>0$, it follows that $\tilde\gamma\in L^1([0,\infty))$. Together with the boundedness of $\tilde\gamma$ from \eqref{eq:bound_tilde_gamma}, this gives $\tilde\gamma^2\in L^1([0,\infty))$. Hence, by \eqref{eq:K2_bound}, $K_2\in L^1([0,\infty))$, and since $\rho\omega_e\in L^1([0,\infty))$, we obtain $d\in L^1([0,\infty))$. Lemma~\ref{lem:dilation}(iii)(a) yields
\begin{equation}\label{eq:delta_iISS}
    |\delta(s)|\le \beta_\delta(|\delta(s_0)|,s-s_0)+\chi\left( \int_{s_0}^s |d(\tau)|{\rm d}\tau \right)
\end{equation}
for some $\beta_\delta\in\mathcal{KL}$ and $\chi\in\mathcal{K_\infty}$, together with $\delta(s)\to0$ as $s\to\infty$. Since $t(s)\to T^-$, we conclude $\delta(t)\to0$ as $t\to T^-$, giving the $\delta$-component of \eqref{eq:rho_t}.

\par\medskip\noindent
\textbf{Global stability bound.} It remains to derive the class-$\mathcal K$ envelope \eqref{eq-stab1}. We bound each component of $(\rho,\delta,\gamma)$ class-$\mathcal K$-ly in $(\rho_0,\delta_0,\gamma_0)$ and combine. The range admits the explicit bound \eqref{eq:bound_rho}, $\rho(t)\le\rho_0+Q_\rho(|\tilde\gamma_0|)$, with $Q_\rho\in\mathcal K$ a polynomial in $|\tilde\gamma_0|$. The reorientation error is class-$\mathcal K$-bounded by Lemma~\ref{lem:capture}(i) applied to the pair $(\rho,|\tilde\gamma|)$ under the dictionary of Part~2:
\begin{equation*}
  |\tilde\gamma(t)|\le\Sigma(|\tilde\gamma_0|,\rho_0)\qquad\text{for all }t\in[0,T),
\end{equation*}
where $\Sigma$ is the envelope of \eqref{eq:lem-sigma}. For the bearing, the iISS estimate \eqref{eq:delta_iISS} evaluated at $s_0=0$ gives
\begin{equation*}
  \boxed{|\delta(t)|\le\beta_\delta(|\delta_0|,0)+\chi\left(\int_0^\infty|d(\tau)|{\rm d}\tau\right),}
\end{equation*}
and the energy on the right is finite and class-$\mathcal K$-bounded in $(\rho_0,|\tilde\gamma_0|)$: Lemma~\ref{lem:dilation}(ii) gives $\int_0^\infty\rho(s)|\omega_{\mathrm e}(s)|\,{\rm d}s\le\omega_M T$, and the filter \eqref{eq:gamma_s_dyn}---driven by the $\rho$-modulated input $\rho a(\delta)\omega_{\mathrm e}$ of $L^\infty$-norm at most $c_0\omega_M$---transfers an $L^1$ bound to $\tilde\gamma$, $\int_0^\infty|\tilde\gamma(s)|\,{\rm d}s\le|\tilde\gamma_0|/c+c_0\omega_M T/c$, and \eqref{eq:K2_bound} controls $K_2$ in $L^1$ accordingly; by the linearly bounded supply rate hypothesis of Lemma~\ref{lem:dilation}(iii)(a) 
(here realized through \eqref{eq:53b} with supply $(2+\varepsilon)v_{\mathrm e}|d|$), the energy $\int_0^\infty |d(\tau)|{\rm d}\tau$ has the class-$\mathcal K$ bound.
Finally, $\gamma=\tilde\gamma-\arctan(\tanh(c_0\delta))$ gives $|\gamma(t)|\le|\tilde\gamma(t)|+c_0|\delta(t)|$, and $|\tilde\gamma_0|\le|\gamma_0|+c_0|\delta_0|$. With $T$ itself bounded by a class-$\mathcal K$ function of $(\rho_0,|\tilde\gamma_0|)$ via \eqref{eq:T1}, combining the three component bounds yields a class-$\mathcal K$ envelope $Q$ such that \eqref{eq-stab1} holds.

\paragraph{4. Boundedness of inputs.}
The estimation \eqref{eq:bound_vp} holds directly from $|\tanh(\cdot)|<1$. For the steering input, using $v_{\mathrm{p}}\sin\gamma - v_{\mathrm{e}}\sin\delta = f_\delta(\delta) + K_2(\delta,\tilde\gamma)$,
\begin{subequations}
\begin{eqnarray}
\omega_{\mathrm{p}}
&=&(1+a(\delta))\frac{f_\delta(\delta)+K_2(\delta,\tilde\gamma)}{\rho}+\left(k+\frac{c}{\rho}\right)\tilde\gamma\\
&=&(1+a(\delta))\left(\frac{f_\delta(\delta)}{\delta}\,q_2+\frac{K_2(\delta,\tilde\gamma)}{\rho}\right)+k\tilde\gamma+c\,q_1, \label{eq:93}
\end{eqnarray}
\end{subequations}
with
\begin{equation}
q_1(s) := \frac{\tilde\gamma(s)}{\rho(s)}, \quad q_2(s) := \frac{\delta(s)}{\rho(s)},
\end{equation}
well defined on $[0,\infty)$ since $\rho(t)>0$ on $[0,T)$. Note also the $\rho$-dynamics in $s$-time,
\begin{equation}\label{eq:rho_dyn_s}
\frac{{\rm d}\rho}{{\rm d}s}=\rho\bigl(f_\rho(\delta)+K_1(\delta,\tilde\gamma)\bigr).
\end{equation}

\par\medskip\noindent
\textbf{Boundedness of $q_1$.} Using \eqref{eq:gamma_s_dyn} and \eqref{eq:rho_dyn_s},
\begin{equation}\label{eq:77}
\frac{{\rm d}q_1}{{\rm d}s}=-\bigl(c+k\rho+f_\rho(\delta)+K_1(\delta,\tilde\gamma)\bigr)q_1-a(\delta)\omega_{\mathrm e}.
\end{equation}
As $s\to\infty$, $\rho,\delta,\tilde\gamma,K_1\to 0$ and $f_\rho\to-\varepsilon v_{\mathrm e}$, so the dominant coefficient $c+k\rho+f_\rho+K_1\to c-\varepsilon v_{\mathrm e}>0$. Hence there exist $\mu_1\in(0,c-\varepsilon v_{\mathrm e})$ and $S_1>0$ with the coefficient $\ge\mu_1$ for $s\ge S_1$. Using $|a|\le c_0$ and $|\omega_{\mathrm e}|\le\omega_M$, \eqref{eq:77} gives $\frac{\rm d}{{\rm d}s}|q_1|\le-\mu_1|q_1|+c_0\omega_M$, and the comparison principle yields $|q_1(s)|\le e^{-\mu_1(s-S_1)}|q_1(S_1)|+c_0\omega_M/\mu_1$ for $s\ge S_1$. On the compact precursor $[0,S_1]$, $q_1$ is bounded by continuity of $\rho,\tilde\gamma$ with $\rho>0$. Therefore $q_1\in L^\infty([0,\infty))$.

\par\medskip\noindent
\textbf{Boundedness of $q_2$.} Using \eqref{eq:rho_dyn_s} and \eqref{eq:51},
\begin{equation}\label{eq:83}
\frac{{\rm d}q_2}{{\rm d}s}=\left(\frac{f_\delta(\delta)}{\delta}-f_\rho(\delta)-K_1(\delta,\tilde\gamma)\right)q_2+\frac{K_2(\delta,\tilde\gamma)}{\rho}-\omega_{\mathrm e},
\end{equation}
with $f_\delta(\delta)/\delta$ extended at $\delta=0$ by $f_\delta'(0)=-((1+\varepsilon)c_0+1)v_{\mathrm e}$. The dominant coefficient tends to $-[(1+\varepsilon)c_0+1-\varepsilon]v_{\mathrm e}\le-2v_{\mathrm e}$ (using $c_0\ge 1$), so there exist $\mu_2\in(0,2v_{\mathrm e})$ and $S_2>0$ with this coefficient $\le-\mu_2$ for $s\ge S_2$. The forcing $K_2/\rho$ is in $L^\infty$ since, by \eqref{eq:K2_bound} and $q_1\in L^\infty$,
\begin{equation}
\frac{|K_2(\delta,\tilde\gamma)|}{\rho}\le(1+\varepsilon)v_{\mathrm e}\left(|q_1|+\frac{\rho q_1^2}{2}\right)\in L^\infty([0,\infty)).
\end{equation}
The same comparison-principle argument as for $q_1$ then gives $q_2\in L^\infty([0,\infty))$. Therefore every term in \eqref{eq:93} is bounded, hence $\omega_{\mathrm p}\in L^\infty([0,T))$. \qed
\end{proof}

\section{Evasion against a maneuvering pursuer}\label{sec:evader}

In this section, we develop a directional evasion law, in the sense of \cite{krstic2026directional}, with a bounded smooth speed control that guarantees sustained evasion from arbitrary initial relative configurations. Let us introduce the error variable
\begin{equation}
    \phi:=\gamma-\pi
\end{equation}
and rewrite \eqref{eq:3} as
\begin{subequations}\label{eq:95}
    \begin{eqnarray}
        \dot{\rho}&=&v_{\mathrm{e}}\cos\delta + v_{\mathrm{p}}\cos\phi \label{eq:95a}\\
        \dot{\phi}&=&-\frac{1}{\rho}(v_{\mathrm{p}} \sin\phi + v_{\mathrm{e}}\sin\delta) - \omega_{\mathrm{p}} \label{eq:95b}\\
        \dot{\delta}&=&-\frac{1}{\rho}(v_{\mathrm{p}} \sin\phi + v_{\mathrm{e}}\sin\delta) - \omega_{\mathrm{e}}. \label{eq:95c}
    \end{eqnarray}
\end{subequations}

\begin{theorem}\label{thm:3}
Consider \eqref{eq:95}. Assume that the pursuer moves with constant speed $v_{\mathrm{p}}>0$ and has a piecewise continuous steering input satisfying
\begin{equation}\label{eq:w_p_bounded}
    |\omega_{\mathrm{p}}(t)|\le \omega_M, \quad \forall t\ge 0.
\end{equation}
for a known $\omega_M>0$. Fix a safety margin $\rho_{\min}>0$.
Let the evader's linear velocity control be 
\begin{equation}\label{eq:ve}
v_{\mathrm{e}} := (1+\varepsilon)\, v_{\mathrm{p}} \sqrt{1+\tanh^2(c_0\phi)},
\end{equation}
and let the evader's steering control be
\begin{subequations}\label{eq:we}
    \begin{eqnarray}
    \omega_{\mathrm{e}}&:=& (a(\phi)-1)\frac{v_{\mathrm{p}}\sin\phi + v_{\mathrm{e}} \sin \delta}{\rho} + \left( k + \frac{c}{\rho - \rho_{\min}} \right)\tilde{\delta} \label{eq:98a}\\
    \tilde{\delta}&:=&\delta-\delta^*(\phi) \label{eq:98b}\\
    \delta^*(\phi)&:=& \arctan(\tanh(c_0\phi)),\label{eq:98c}\\
    a(\phi)&:=&\frac{c_0\sech^2(c_0\phi)}{1 + \tanh^2(c_0\phi)}\in(0,c_0], \label{eq:a_func}
    \end{eqnarray}
\end{subequations}
where $\varepsilon>0$, $c_0\ge c_0^\star(\varepsilon):=\ln((2+\varepsilon)/\varepsilon)/(2\pi)$, $c>(3+2\varepsilon) v_{\mathrm{p}}$. Let $\eta>0$ satisfy \eqref{eq:eta} and choose $k>c_0\omega_M/\eta$.
Then, for every initial condition $\rho_0>\rho_{\min}$, $\delta_0\in\mathbb{R}$, and $\phi_0\in\mathbb{R}$, the closed-loop solution exists uniquely for all $t\ge 0$ and satisfies $\rho(t)>\rho_{\min}$ for all $t\ge 0$.
Moreover, if for some $\lambda\ge 0$,
\begin{equation}\label{eq:t_omega_p}
\lim_{t\to\infty} t^{1+\lambda}\,\omega_{\mathrm{p}}(t)=0,
\end{equation}
then the spinaway objective is achieved:
\begin{equation}
\lim_{t\to\infty}\delta(t)=0
\quad \text{and} \quad
\lim_{t\to\infty}\phi(t)=0.
\end{equation}
Finally, the speed is uniformly bounded as
\begin{equation}\label{eq:bound_ve}
    (1+\varepsilon)v_\mathrm{p}\le v_\mathrm{e}(t)< \sqrt{2}(1+\varepsilon)v_\mathrm{p}, \quad \forall t\ge 0,
\end{equation}
and the steering input satisfies $\omega_\mathrm{e}\in L^\infty([0,\infty))$.

\end{theorem}

\begin{proof}[Theorem \ref{thm:3}]~

\paragraph{1. Closed-loop system and safety margin.}

Parallel to the proof of Theorem~\ref{thm:2}, define 
\begin{subequations}
\begin{eqnarray}
\Psi_\rho(\phi)
&:=&
v_{\mathrm p}(1+\varepsilon+\cos\phi),\\
\Psi_\phi(\phi)
&:=&
-v_{\mathrm p}\left((1+\varepsilon)\tanh(c_0\phi)+\sin\phi\right),\\
\Delta_\rho(\phi,\tilde\delta)
&:=&
(1+\varepsilon)v_{\mathrm p}
\left[
(\cos\tilde\delta-1)-\tanh(c_0\phi)\sin\tilde\delta
\right],\\
\Delta_\phi(\phi,\tilde\delta)
&:=&
(1+\varepsilon)v_{\mathrm p}
\left[
\tanh(c_0\phi)(1-\cos\tilde\delta)-\sin\tilde\delta
\right].
\end{eqnarray}
\end{subequations}
Substituting \eqref{eq:ve}, \eqref{eq:98b}, and \eqref{eq:98c} into \eqref{eq:95a}-\eqref{eq:95b}, and using the identities \eqref{eq:cossin}, the closed-loop $\rho$- and $\phi$-dynamics can be written as
\begin{subequations}
\begin{eqnarray}
\dot\rho
&=&
\Psi_\rho(\phi)+\Delta_\rho(\phi,\tilde\delta), \label{eq:rho_evader}\\
\dot\phi
&=&
\frac{1}{\rho}
\left[
\Psi_\phi(\phi)+\Delta_\phi(\phi,\tilde\delta)
\right]
-\omega_{\mathrm p},
\end{eqnarray}
\end{subequations}
where the residual terms satisfy
\begin{equation}
    \max\left\{|\Delta_\rho(\phi,\tilde\delta)|,|\Delta_\phi(\phi,\tilde\delta)|\right\}\le (1+\varepsilon)v_{\mathrm p}
\left(|\tilde\delta|+\frac{\tilde\delta^2}{2}\right). \label{eq:Delta_rho}
\end{equation}

Let $r(t):=\rho(t)-\rho_{\min}$.
Then $r(0)>0$, and the feedback law is locally Lipschitz on the open domain $\rho>\rho_{\min}$, hence the closed-loop solution exists and is unique on a maximal interval $[0,T_{\max})$ on which $r(t)>0$.
On $[0,T_{\max})$, differentiating $\tilde{\delta}$ along the closed-loop trajectories gives
\begin{equation}
\dot{\tilde\delta}
=
-\left(k+\frac{c}{r}\right)\tilde\delta
+
a(\phi)\omega_{\mathrm p}
\end{equation}
Since $a(\phi)\in(0,c_0]$, the comparison principle yields the ISS estimate
\begin{equation}\label{eq:ISS_tilde_delta}
\boxed{
|\tilde\delta(t)|
\le
e^{-kt}|\tilde\delta_0|
+
c_0\int_{0}^{t}
e^{-k(t-\tau)}
|\omega_{\mathrm p}(\tau)|\,{\rm d}\tau,
\quad t\in[0,T_{\max}) .}
\end{equation}
We now invoke Lemma~\ref{lem:safety} (Appendix~\ref{app:safety}) on $[0,T_{\max})$ with margin $r=\rho-\rho_{\min}$ and error $e=\tilde\delta$. Hypothesis \eqref{eq:saf-e} holds with $B=c_0\omega_M$: since $|a(\phi)|\le c_0$ and $|\omega_{\mathrm p}|\le\omega_M$, the $\tilde\delta$-dynamics give ${\rm d}|\tilde\delta|/{{\rm d}t}\le-(k+\frac{c}{r})|\tilde\delta|+c_0\omega_M$. Hypothesis \eqref{eq:saf-r} holds with $\mu=\varepsilon v_{\mathrm p}$, $\kappa=(1+\varepsilon)v_{\mathrm p}$, and $L=(3+2\varepsilon)v_{\mathrm p}$: from \eqref{eq:rho_evader}, $\dot r=\Psi_\rho(\phi)+\Delta_\rho(\phi,\tilde\delta)$, where $\Psi_\rho\ge\varepsilon v_{\mathrm p}$ and, by \eqref{eq:Delta_rho}, $|\Delta_\rho|\le(1+\varepsilon)v_{\mathrm p}(|\tilde\delta|+\tilde\delta^2/2)$, so that $\dot r\ge\varepsilon v_{\mathrm p}-(1+\varepsilon)v_{\mathrm p}(|\tilde\delta|+\tilde\delta^2/2)$, while $\Psi_\rho+\Delta_\rho\ge-(3+2\varepsilon)v_{\mathrm p}$ gives $\dot r\ge-L$; the requirement $c>L$ is the standing assumption $c>(3+2\varepsilon)v_{\mathrm p}$. The gain condition reads $c_0\omega_M/k<\eta$, i.e., $k>c_0\omega_M/\eta$, with $e^\star=\eta$ fixed by \eqref{eq:eta}.

Lemma~\ref{lem:safety} therefore applies on $[0,T_{\max})$. With $\nu:=c-(3+2\varepsilon)v_{\mathrm p}$, $\underline{r}:=\min\{r(0),r_\star\}$,
\begin{equation*}
\zeta_M:=\max\left\{\frac{|\tilde\delta_0|}{\rho_0-\rho_{\min}},\,\frac{c_0\omega_M}{\nu}\right\},\quad
r_\star:=\frac{1}{\zeta_M}\left(\sqrt{1+\frac{\varepsilon}{1+\varepsilon}}-1\right),\quad
T_\eta:=\frac{1}{k}\max\left\{0,\ln\frac{|\tilde\delta_0|}{\eta-c_0\omega_M/k}\right\},
\end{equation*}
Lemma~\ref{lem:safety}(ii) gives the bounded ratio $|\tilde\delta(t)|\le\zeta_{M}\,r(t)$ on $[0,T_{\max})$, and Lemma~\ref{lem:safety}(iii) gives $r(t)\ge\underline r>0$ on $[0,T_{\max})$. Hence $\rho(t)\ge\rho_{\min}+\underline r>\rho_{\min}$ throughout $[0,T_{\max})$, and since the other state components $(\delta,\phi)$ are bounded on bounded $t$-intervals (the closed-loop right-hand side is bounded whenever $r$ is bounded away from $0$), no finite-time blow-up occurs at $T_{\max}$; therefore $T_{\max}=\infty$. The linear escape bound in Lemma~\ref{lem:safety}(iii) then gives the explicit safety bound on all of $[0,\infty)$:
\begin{equation}\label{eq:safety}
\boxed{\rho(t)\ge\underline\rho(t):=\left\{
\begin{aligned}
&\rho_{\min}+\underline{r},
& 0\le t<T_\eta,\\
&\rho_{\min}+\underline{r}+\dfrac{\varepsilon v_{\mathrm p}}{2}(t-T_\eta),
& t\ge T_\eta .
\end{aligned}
\right.
}
\end{equation}
This shows that the range is uniformly separated from $\rho_{\min}$ for all time and grows at least linearly after the finite reorientation transient.

\paragraph{2. Asymptotic spinaway property.}

Following \eqref{eq:t_omega_p}, the ISS estimate \eqref{eq:ISS_tilde_delta} implies $\tilde\delta(t)\to 0$.
Indeed, \eqref{eq:ISS_tilde_delta} is a convolution estimate with an exponentially decaying kernel, and hence the response of the exponentially stable $\tilde\delta$-subsystem to a vanishing input also vanishes. Consequently, since $\Delta_\phi(\phi,\tilde\delta)$ is globally bounded in $\phi$ and vanishes uniformly with respect to $\phi$ as $\tilde\delta\to0$, one has $\Delta_\phi(\phi(t),\tilde\delta(t))\to 0$.

We obtain $\phi(t)\to0$ from Lemma~\ref{lem:dilation} (Appendix~\ref{app:dilation}), applied to the $\phi$-subsystem written as $\dot\phi=\left(\Psi_\phi(\phi)+p_\phi(t)\right)/\rho$ with $f=\Psi_\phi$ and disturbance $p_\phi:=\Delta_\phi(\phi,\tilde\delta)-\rho\,\omega_{\mathrm p}$. Its hypotheses hold:

\par\medskip\noindent
\textbf{Exhausted clock (diverging regime).} Since $\Psi_\rho$ and $\Delta_\rho$ are globally bounded, so is $\dot\rho$, and $\rho(t)\le\rho_0+R_\rho t$ with $R_\rho:=(3+2\varepsilon)v_{\mathrm p}$; the range thus diverges at most linearly, so $\int_0^\infty\rho^{-1}(t){\rm d}t=\infty$, and Lemma~\ref{lem:dilation}(i) provides the stretched time $s(t)=\int_0^t\rho^{-1}(\tau){\rm d}\tau\to\infty$ and the desingularized dynamics $\frac{{\rm d}\phi}{{\rm d}s}=\Psi_\phi(\phi)+d_\phi(s)$, where $d_\phi(s):=p_\phi(t(s))$.

\par\medskip\noindent
\textbf{Vanishing input.} The term $\Delta_\phi(\phi,\tilde\delta)\to0$, as shown above. For the other term, using $\rho(t)\le\rho_0+R_\rho t$ and \eqref{eq:t_omega_p},
\begin{equation*}
|\rho(t)\omega_{\mathrm p}(t)|\le\left(\rho_0+R_\rho t\right)|\omega_{\mathrm p}(t)|
=\left(\frac{\rho_0}{t^{1+\lambda}}+\frac{R_\rho}{t^{\lambda}}\right)\big|t^{1+\lambda}\omega_{\mathrm p}(t)\big|\to0 .
\end{equation*}
Hence $d_\phi(s)\to0$ as $s\to\infty$.

\par\medskip\noindent
\textbf{Sign and liminf at infinity for $\Psi_\phi$.} By Lemma~\ref{lem:1}, $\phi\,\Psi_\phi(\phi)<0$ for $\phi\ne0$ and $\liminf_{|\phi|\to\infty}|\Psi_\phi(\phi)|\ge\varepsilon v_{\mathrm p}>0$.

Lemma~\ref{lem:dilation}(iii)(b) therefore applies and gives $\phi(t)\to0$ as $t\to\infty$.
Finally, since $\tilde\delta(t)\to0$ and $\delta^*(\phi)=\arctan(\tanh(c_0\phi))\to0$ as $\phi\to0$, we conclude that $\delta(t) \to 0$ as $t\to\infty$.

\paragraph{3. Boundedness of inputs.}
The estimation \eqref{eq:bound_ve} holds directly from the fact that $|\tanh(\cdot)|<1$. For the boundedness of the steering input $\omega_{\mathrm{e}}$, the first term of \eqref{eq:98a} is bounded because $\rho(t)\ge\rho_{\min}+\underline{r}$ and $v_{\mathrm e}$ is bounded. The second term is bounded because $\tilde\delta$ is bounded and $\zeta:={\tilde\delta}/{(\rho-\rho_{\min})}\in L^\infty$. 
Therefore, $\omega_{\mathrm e}\in L^\infty([0,\infty))$. \qed

\end{proof}

\section{Simulation results}\label{sec:simulation}

The simulations illustrate the transient mechanisms predicted by the analysis: reorientation from unfavorable initial headings, possible nonmonotonicity of the range, finite-time capture in pursuit, and safety-margin preservation in evasion.
The gains are selected to make the transient mechanisms visible. Higher gains produce faster but less informative trajectories. Angular variables are plotted as continuous representatives, obtained by unwrapping. All quantities are reported in SI units.

\subsection{Directional pursuit of a maneuvering evader}

\begin{figure}
    \centering
    \begin{subfigure}[t]{0.95\textwidth}
        \centering
        \includegraphics[height=0.25\textheight, width=\linewidth, keepaspectratio]{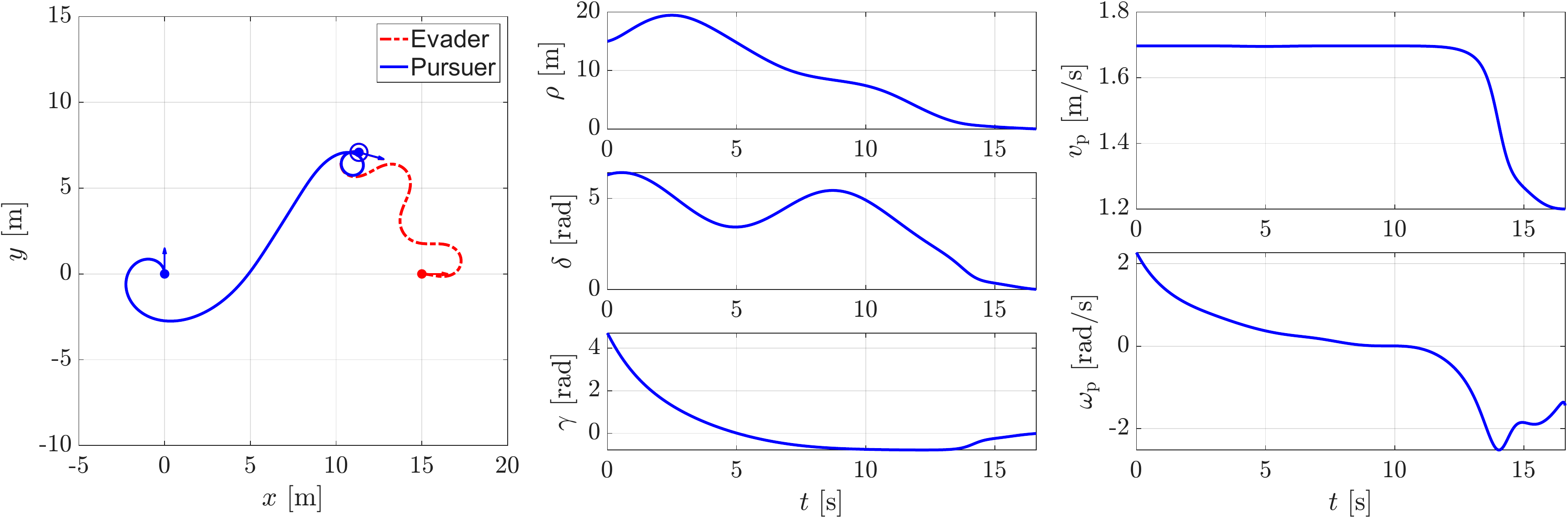}
        \caption{Case 1: Unfavorable initial heading condition, with $\gamma_0=-\pi/2$ and $\omega_{\mathrm e}(t)=1.4\sin(0.7t)-0.6\cos(0.5t)$.}
        \label{fig:pursuit-2}
        \vspace{0.2cm}
    \end{subfigure}

    \begin{subfigure}[t]{0.95\textwidth}
        \centering
        \includegraphics[height=0.25\textheight, width=\linewidth, keepaspectratio]{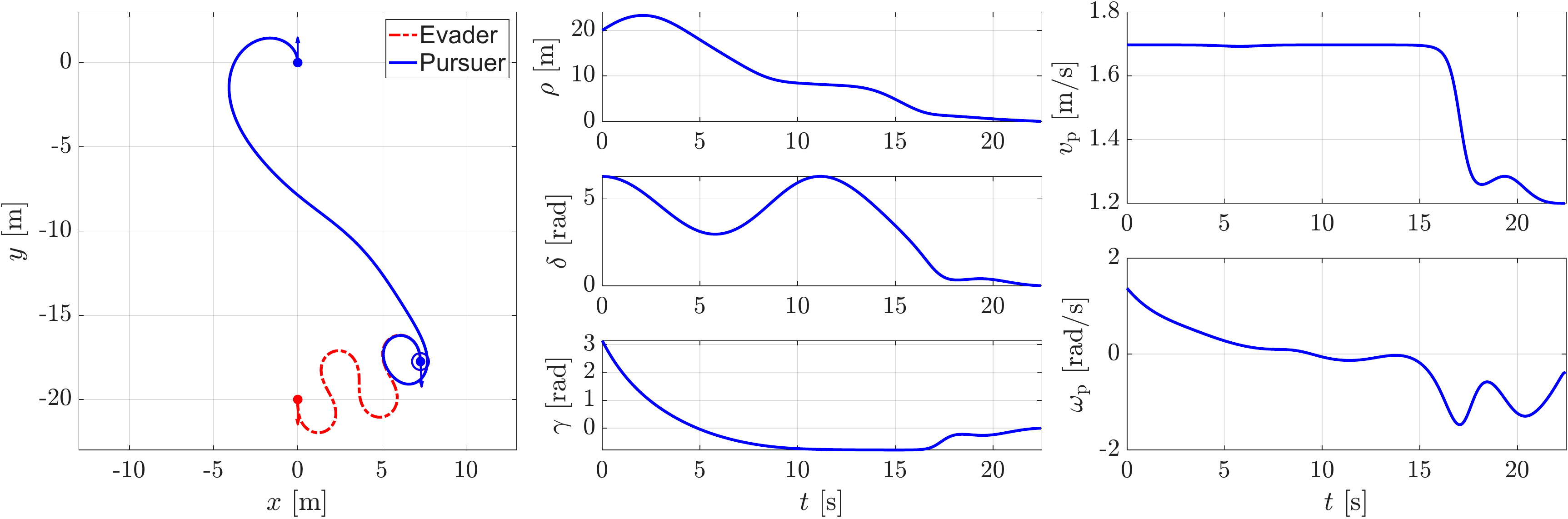}
        \caption{Case 2: Unfavorable initial heading condition, with $\gamma_0=\pi$ and $\omega_{\mathrm e}(t)=\sin(0.55t)$.}
        \label{fig:pursuit-3}
        \vspace{0.2cm}
    \end{subfigure}
    \caption{Simulations illustrating the directional pursuit control in Theorem~\ref{thm:2}. The circles indicate the capture location.}
    \label{fig:pursuit_simulation}
\end{figure}

We illustrate the pursuit controller of Theorem~\ref{thm:2} in configurations that violate the favorable initial-heading condition\footnote{Following \cite{krstic2026directional}, a favorable pursuit initial condition means $\cos\gamma_0>0$, and a favorable evasion initial condition means \(\cos\delta_0>0\).} of the half-global result in \cite{krstic2026directional}. The evader moves at $v_e=1$ with $|\omega_e|\le \omega_M=2$. The pursuer uses $\varepsilon=0.2$, $c_0=1$, $c=5$, and $k=0.1$. These gains are chosen to display the reorientation transient clearly.

The pursuer starts at $(0,0,\pi/2)$. In Case 1, the evader starts at $(15,0,0)$ with $\omega_e(t)=1.4\sin(0.7t)-0.6\cos(0.5t)$, giving $\gamma_0=-\pi/2$. In Case 2, the evader starts at $(0,-20,-\pi/2)$ with $\omega_e(t)=\sin(0.55t)$, giving the strongly unfavorable value $\gamma_0=\pi$. Figure 2 shows that in both cases the pursuer reorients, the range $\rho$ exhibits the predicted nonmonotone transient, and finite-time capture is achieved with $\delta,\gamma\to0$.

\subsection{Evasion against a maneuvering pursuer}

\begin{figure}
    \centering
    \begin{subfigure}[t]{0.95\textwidth}
        \centering
        \includegraphics[height=0.25\textheight, width=\linewidth, keepaspectratio]{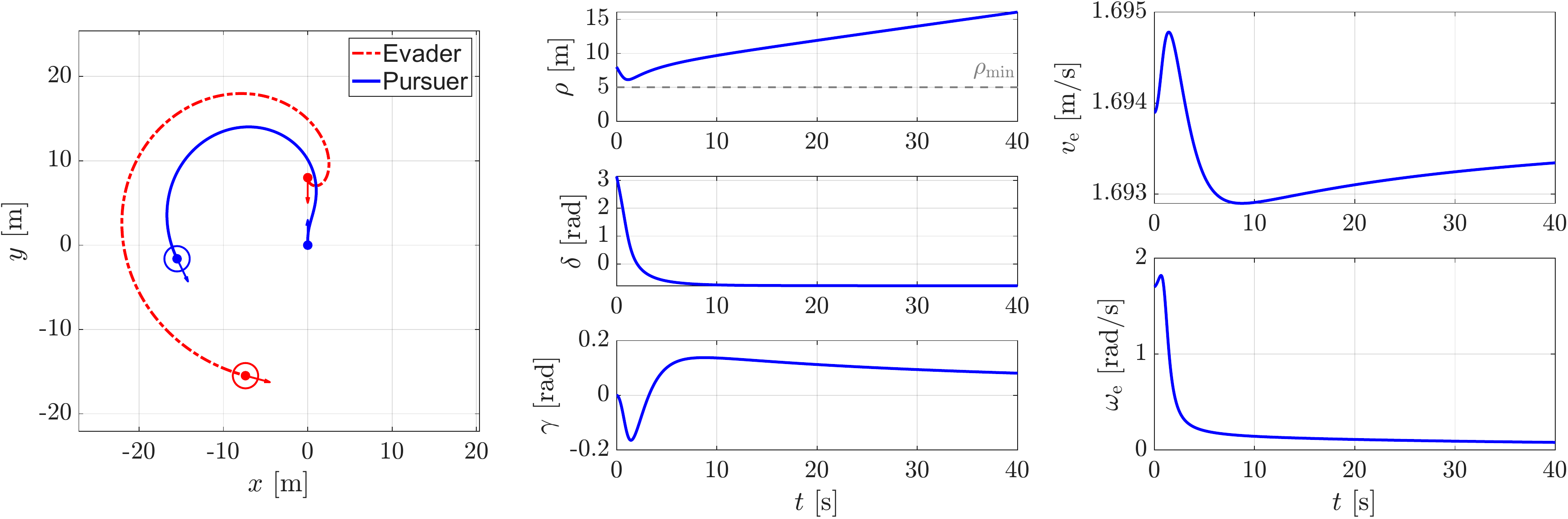}
        \caption{Case 1: Strongly unfavorable initial configuration with the pursuer using \(\omega_{\mathrm p}=\tanh(\gamma)\).}
        \label{fig:evasion-1}
        \vspace{0.2cm}
    \end{subfigure}

    \begin{subfigure}[t]{0.95\textwidth}
        \centering
        \includegraphics[height=0.25\textheight, width=\linewidth, keepaspectratio]{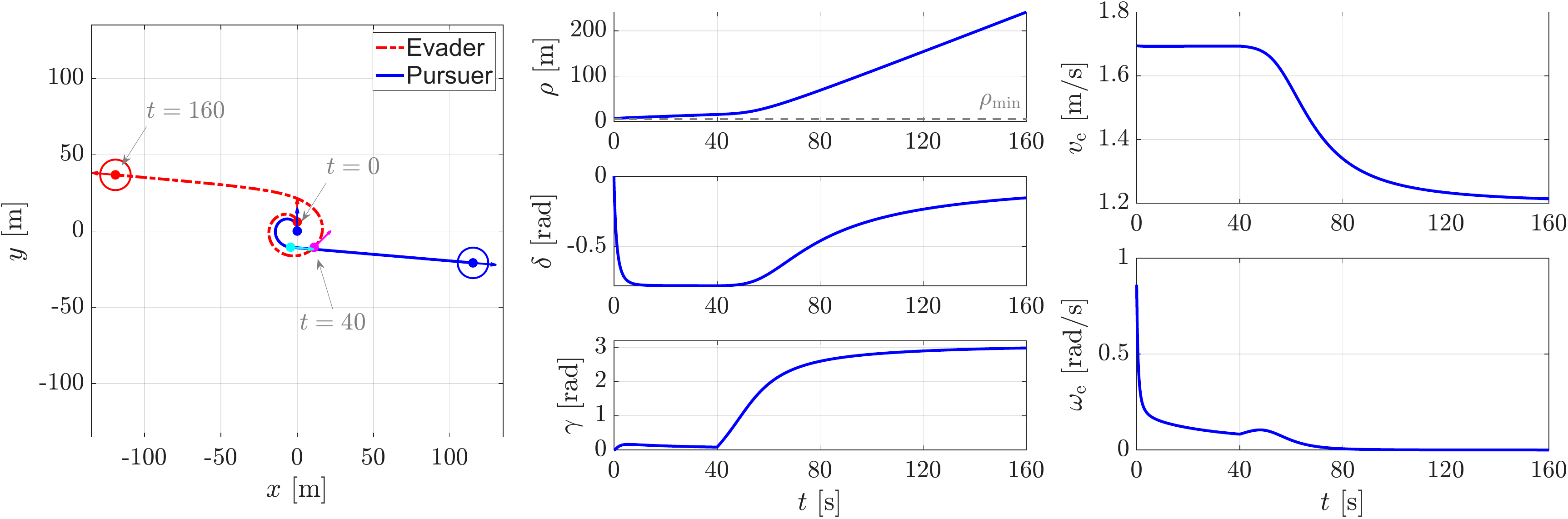}
        \caption{Case~2: The pursuer uses \(\omega_{\mathrm p}=\tanh(\gamma)\) for \(t\le40\) and \(\omega_{\mathrm p}=0\) for \(t>40\). The magenta and cyan arrows mark the positions of the evader and pursuer, respectively, at \(t=40\).}
        \label{fig:evasion-2}
        \vspace{0.2cm}
    \end{subfigure}
    \caption{Simulations illustrating the directional evasion control law in Theorem~\ref{thm:3}. 
The evader is shown in red and the pursuer in blue; dots denote initial positions and circles denote terminal positions over the simulated horizon. 
The dashed horizontal line in the \(\rho\)-plots denotes \(\rho_{\min}=5\). }
    \label{fig:evasion_simulation}
\end{figure}

We next illustrate the evasion controller of Theorem~\ref{thm:3}. The pursuer moves at $v_p=1$ with $|\omega_p|\le \omega_M=1$, and the prescribed safety margin is $\rho_{\min}=5$. The evader uses $\varepsilon=0.2$, $c_0=1$, $c=1$, and $k=0.1$. The gains are selected to make the safety transient visible and are below the conservative sufficient ranges used in the proof.

In Case 1, the evader starts at $(0,8,-\pi/2)$, the pursuer starts at $(0,0,\pi/2)$, and $\omega_p(t)=\tanh(\gamma(t))$, giving strongly unfavorable initial heading condition $\delta_0=\pi$. Since the pursuer keeps maneuvering, the decay condition required for spinaway is not imposed in this case. Figure 3(a) shows that the range initially approaches the safety margin, remains above $\rho_{\min}$, and then increases.

In Case 2, the evader starts at $(0,6,\pi/2)$, the pursuer starts at $(0,0,\pi/2)$, and the initial range $\rho_0=6$ is close to $\rho_{\min}=5$. The pursuer uses $\omega_p(t)=\tanh(\gamma(t))$ for $t\le40$ and $\omega_p(t)=0$ for $t>40$. Figure 3(b) shows that the safety margin is preserved throughout the maneuver; after the pursuer stops turning, the range grows and the angular variables converge to the spinaway configuration, $\delta\to0$ and $\gamma\to\pi$.

\section{Conclusion}\label{sec:conclusion}

This paper extends the directional pursuit-evasion designs of \cite{krstic2026directional} from favorable initial headings to arbitrary initial relative configurations, at the price of requiring a known uniform bound on the opponent’s steering input. The pursuit law guarantees finite-time capture with directional alignment, while the evasion law guarantees separation from a prescribed safety margin and achieves spinaway under an additional decay condition on the pursuer’s steering input.

The robustness mechanism shifts from ISS to iISS, which is captured by abstract auxiliary lemmas for finite-time coextinction, disturbance budgeting, and safety-margin persistence. These results clarify the transient cost of globality and may be useful for other control problems. Future work will address actuator constraints, and extensions to dynamic and multi-agent models.

\appendix
\setcounter{section}{0}
\renewcommand{\thesection}{\Alph{section}}

\newcommand{\appsection}[1]{%
  \refstepcounter{section}%
  \section*{Appendix \thesection.\ #1}%
}

\appsection{An elementary positivity lemma}\label{app:lemma1}

\begin{lemma}\label{lem:1}
Let $\varepsilon>0$ and define
\begin{equation}
    c_0^\star(\varepsilon):=\frac{1}{2\pi}\ln \frac{2+\varepsilon}{\varepsilon}.
\end{equation}
For any $c_0\ge c_0^\star(\varepsilon)$, define $g(x):=(1+\varepsilon)\tanh(c_0x)+\sin(x)$. Then, the function $x\mapsto g(x)^2$ is positive definite. Moreover, $\liminf_{|x|\to\infty}|g(x)|=\varepsilon$.
\end{lemma}

\begin{proof}
By oddness, it suffices to show $g(x)>0$ on $(0,\infty)$. The choice of $c_0^\star$ gives $(1+\varepsilon)\tanh(c_0\pi)\ge 1$. On $(0,\pi]$, both summands are nonnegative with $\tanh(c_0x)>0$, so $g(x)>0$. On $[\pi,\infty)$, monotonicity of $\tanh$ gives $(1+\varepsilon)\tanh(c_0x)\ge 1$, so $g(x)\ge 1+\sin x\ge 0$, with strict inequality everywhere outside the zeros of $1+\sin x$ (which lie strictly inside this interval, where $\tanh$ is strictly larger). The gap follows from $(1+\varepsilon)\tanh(c_0x)\to 1+\varepsilon$ as $x\to\infty$ and oddness. \qed
\end{proof}

\appsection{Finite-time coextinction of two mutualistic populations}\label{app:capture}

Read $(\rho,w)$ as a cost-free mutualism: $w$ supports $\rho$ through $\chi(w)$, $\rho$ shelters $w$ by suppressing the mortality $k+c/\rho$. Despite this, an intrinsic decay $-q$ in $\rho$ and a singular mortality $c/\rho$ in $w$ force finite-time coextinction (Fig.~\ref{fig:coextinction}); in Theorem~\ref{thm:2}, $\rho$ is the range and $w=|\tilde\gamma|$ the reorientation error. 

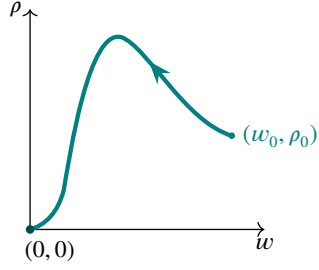
\begin{figure}
\centering
\begin{tikzpicture}[scale=0.62]
  \draw[->,line width=0.5pt] (0,0) -- (5.0,0) node[anchor=north,inner sep=3pt] {$w$};
  \draw[->,line width=0.5pt] (0,0) -- (0,4.7) node[anchor=east,inner sep=3pt] {$\rho$};
  \draw[line width=1.6pt,teal,
        postaction={decorate},
        decoration={markings,
          mark=at position 0.3 with {\arrow{Stealth[length=3mm]}}}]
    (4.3,2.0) .. controls (3.3,2.3) and (2.5,3.9) ..
    (1.95,4.1) .. controls (1.25,4.3) and (0.85,1.6) ..
    (0.72,0.85) .. controls (0.6,0.36) and (0.32,0.09) .. (0,0);
  \fill[teal] (4.3,2.0) circle (2pt) node[anchor=west,inner sep=4pt] {$(w_0,\rho_0)$};
  \fill[teal!60!black] (0,0) circle (2.6pt);
  \node[anchor=north west,inner sep=4pt] at (-0.35,0) {$(0,0)$};
\end{tikzpicture}
\caption{A representative trajectory of the pair $(w,\rho)$ in Lemma~\ref{lem:capture}: starting from $(w_0,\rho_0)$, $\rho$ overshoots and then decays to zero in finite time while $w$ decreases monotonically, the trajectory ending at the origin (coextinction). The arrowhead marks the direction of flow.}
\label{fig:coextinction}
\end{figure}

\begin{lemma}[Global finite-time stability of two mutualistic populations]\label{lem:capture}
Let $w,\rho:[0,T)\to [0,\infty)$ be absolutely continuous, with $\rho>0$ on $[0,T)$, where $[0,T)$ is the maximal interval of positivity of $\rho$. 
Suppose that there exist constants $k,c,q,L>0$, $b\ge0$, and a function $\chi\in\mathcal{K}$ satisfying $\sup_{r\ge 0}\chi(r)>q$, such that, almost everywhere,
\begin{subequations}
    \begin{eqnarray}
  \dot w &\le& -\left(k+\frac{c}{\rho}\right)w+b, \label{eq:lem-w}\\
  -L &\le& \dot\rho \;\le\; -q+\chi(w). \label{eq:lem-rho}
\end{eqnarray}
\end{subequations}
Assume the gain condition $\underline w:=b/k<\chi^{-1}(q/2)=:w^\star$.
Denote $w_0:=w(0)$ and $\rho_0:=\rho(0)$. Define $\alpha:={c}/{L}$,
\begin{equation}\label{eq:lem-data}
  T_1:=\max\left\{0,\frac{1}{k}\ln\frac{\max\{w_0,\underline w\}-\underline w}{w^\star-\underline w}\right\},
  \quad
  \Omega(w_0):=\left\{\begin{aligned}
      &0, & w_0\le\chi^{-1}(q),\\
      &\frac{1}{k}\int_{\chi^{-1}(q)}^{w_0}\frac{\chi(u)-q}{u-\underline w}\,{\rm d}u,& w_0>\chi^{-1}(q),
  \end{aligned}  \right.
\end{equation}
\vspace{-0.2cm}
\begin{equation}\label{eq:lem-sigma}
  \Sigma(w_0,\rho_0):=\max\left\{w_0,\,\frac{b\left(\rho_0+\Omega(w_0)\right)}{c+k\left(\rho_0+\Omega(w_0)\right)}\right\},
\end{equation}
with the convention $\ln(0):=-\infty$, where $\Omega$ and $\Sigma$ are continuous and nondecreasing in their arguments satisfying $\Omega(0)=0$ and $\Sigma(0,0)=0$. Then,
\begin{itemize}
  \item[(i)] $w(t)\le\Sigma(w_0,\rho_0)$ for all $t\in[0,T)$;
  \item[(ii)] $w(t)\to0$ as $t\to T^-$;
  \item[(iii)] $\rho(t)\le\rho_0+\Omega(w_0)$ for all $t\in[0,T)$;
  \item[(iv)] $\rho(t)\to0$ as $t\to T^-$, 
\end{itemize}
where 
\begin{equation}\label{ineq-T}
T\le T_1+\dfrac{2}{q}\left(\rho_0+\Omega(w_0)\right)<\infty\,.
\end{equation}
\end{lemma}

\noindent
\textit{Interpretation: Researcher and thankless topic.}
Read $\rho$ as a researcher's standing and $w$ as the effort he (she) spends on a thankless research topic. Each supports the other at no cost: the researcher's effort sustains his standing through $\chi(w)$, while his standing eases the burden $k+c/\rho$ on his effort. But standing also decays on its own at rate $q$, and once effort drops below the domination threshold, this decay can no longer be arrested. The researcher's standing reaches zero in finite time; when it does, $c/\rho\to\infty$ stops shelter his from his effort feeling burdensome, and his effort collapses with his standing.

\begin{proof}

\emph{Step 1 (Auxiliary bound).} Since $c/\rho>0$,  \eqref{eq:lem-w} gives $\dot w\le-kw+b$, and therefore, by the comparison principle,
\begin{equation*}
  w(t)\le\bar w(t):=\underline w+e^{-kt}\left(w_0-\underline w\right)\le\max\{w_0,\underline w\},\quad t\in[0,T).
\end{equation*}
The gain condition gives $\underline w<w^\star$. By the definition of $T_1$, the comparison trajectory $\bar{w}(t)$ enters the interval $[0,w^\star]$ no later than $T_1$, and remains there afterward. Therefore, $w(t)\le w^\star$ for all $t\ge T_1$.

\emph{Step 2 (Range bound and finite-time extinction of $\rho$).} From Step~1, for $t\ge T_1$ we have $\chi(w(t))\le\chi(w^\star)=q/2$, and \eqref{eq:lem-rho} yields $\dot\rho(t)\le-q/2$ for all $t\ge T_1$.
We next estimate the possible excursion of $\rho$ before this uniform decay phase. Since $\chi$ is increasing and $w(t)\le\bar w(t)$,
\begin{equation*}
  \rho(t)-\rho_0\le\int_0^t\left(\chi(\bar w(\tau))-q\right)\,{\rm d}\tau
  \le\int_0^\infty \max\left\{0,\chi(\bar w(\tau))-q\right\}\,{\rm d}\tau .
\end{equation*}
The integrand is positive only when $\bar w(\tau)>\chi^{-1}(q)$. If $w_0\le \chi^{-1}(q)$, the last integral is zero. If $w_0>\chi^{-1}(q)$, then $\bar w$ decreases through this interval and, using
$\dot{\bar w}=-k(\bar w-\underline w)$, the change of variables $u=\bar w(\tau)$ gives
\begin{equation*}
  \int_0^\infty \max\left\{\chi(\bar w(\tau))-q\right\}\,{\rm d}\tau
  =\frac1k\int_{\chi^{-1}(q)}^{w_0}\frac{\chi(u)-q}{u-\underline w}\,{\rm d}u=\Omega(w_0).
\end{equation*}
Thus, $\rho(t)\le\rho_0+\Omega(w_0)$ on $[0,T_1)$, while for $t\ge T_1$ the decay $\dot\rho\le-q/2$ keeps $\rho(t)\le\rho(T_1)\le\rho_0+\Omega(w_0)$; hence $\rho(t)\le\rho_0+\Omega(w_0)$ for all $t\in[0,T)$, which is~(iii). 

If $T>T_1$, the estimate $\dot\rho\le -q/2$ on $[T_1,T)$ gives
\begin{equation*}
    T-T_1\le \frac{2}{q}\rho(T_1)\le \frac{2}{q}(\rho_0+\Omega(w_0)).
\end{equation*}
If $T\le T_1$, this bound is trivial. Hence, \eqref{ineq-T} holds. Since $[0,T)$ is the maximal interval on which $\rho>0$, and since $\dot\rho$ is bounded on $[0,T)$ by the preceding bounds, $\rho(t)$ has a finite limit as $t\to T^-$. This limit cannot be positive; otherwise positivity could be extended beyond $T$. Therefore, $\rho(t)\to0$ as $t\to T^-$, which is~(iv).

\emph{Step 3 (Bound and coextinction of $w$).} Define $\lambda(t):=k+c/\rho(t)$ and $\Lambda(t):=\int_0^t\lambda(s){\rm d}s$. Variation-of-constants gives
\begin{equation*}
  w(t)\le w_0 e^{-\Lambda(t)}+\int_0^tb\,e^{-(\Lambda(t)-\Lambda(s))}{\rm d}s\le w_0 e^{-\Lambda(t)}+ \frac{b(\rho_0+\Omega(w_0))}{c+k(\rho_0+\Omega(w_0))}(1-e^{-\Lambda(t)}).
\end{equation*}
With $w_0\,e^{-\Lambda(t)}\le w_0$, this yields $w(t)\le\Sigma(w_0,\rho_0)$, proving~(i).

From $\dot\rho\ge-L$ and $\rho(T)=0$, $\rho(t)\le L(T-t)$, so $\lambda(t)\ge k+\alpha/(T-t)$. Substituting the sharper lower bound on $\lambda$ into the same variation-of-constants estimate,
\begin{equation*}
 w(t)\le w_0 e^{-kt}\left(\frac{T-t}{T}\right)^{\alpha}
        +b\int_0^t e^{-k(t-s)} \left(\frac{T-t}{T-s}\right)^{\alpha}\,{\rm d}s.
\end{equation*}
The first term vanishes; for the integral term, with $r=T-t$, $u=T-s$,
\begin{equation*}
  \int_0^t e^{-k(t-s)} \left(\frac{T-t}{T-s}\right)^{\alpha}{\rm d}s = r^\alpha\int_r^T e^{-k(u-r)}u^{-\alpha}{\rm d}u,
\end{equation*}
which is $\mathcal{O}(r^\alpha)$ if $\alpha<1$, $\mathcal{O}(r\ln(1/r))$ if $\alpha=1$, and $\mathcal{O}(r)$ if $\alpha>1$. Hence $w(t)\to 0$ as $t\to T^-$, proving~(ii). \qed
\end{proof}

In role and spirit, Lemma~\ref{lem:capture} is akin to Lemma~4.1 of \cite{karafyllis2025robust}, where a predator-prey structure was used to analyze stability outside a small-gain configuration.

\appsection{Endogenous time dilation and disturbance budget}\label{app:dilation}

A subsystem driven at the singular rate $1/\rho$ becomes regular under the change of time $s=\int_0^t \rho^{-1}$, which is a bijection of $[0,T)$ onto $[0,\infty)$ whenever $\int_0^T\rho^{-1}=\infty$. This occurs, for example, in the finite-time capture case when $\rho\to0$, and also in infinite-time spinaway regimes when $\rho\to\infty$ slowly enough, such as linear growth. In the finite-time case, $\rho$ has total mass $T$ in $s$, so any $\rho$-modulated bounded perturbation is integrable in the stretched time $s$.

\begin{lemma}[Endogenous-clock time dilation and disturbance budget]\label{lem:dilation}
Let $T\in(0,\infty]$, and let $\rho:[0,T)\to(0,\infty)$ be absolutely continuous and satisfy $\int_0^{T}\rho(\tau)^{-1}\,{\rm d}\tau=\infty$. Let $\delta:[0,T)\to\mathbb{R}$ satisfy, almost everywhere,
\begin{equation}\label{eq:dil-sys}
  \dot\delta=\frac{1}{\rho}\left(f(\delta)+p(t)\right),
\end{equation}
where $f:\mathbb{R}\to\mathbb{R}$ is locally Lipschitz and $f(0)=0$, and $p:[0,T)\to\mathbb{R}$. Define $s(t):=\int_0^t\rho(\tau)^{-1}\,{\rm d}\tau$. Then:
\begin{enumerate}
  \item[(i)] \emph{(Reduction)} The map $s$ is a strictly increasing bijection of $[0,T)$ onto $[0,\infty)$. In the $s$-variable, the singularly scaled system becomes the regular system
  \begin{equation}\label{eq:dil-desing}
    \frac{{\rm d}\delta}{{\rm d}s}=f(\delta)+d(s),\quad d(s):=p(t(s)),
  \end{equation}
  The hypothesis $\int_0^T\rho^{-1}(\tau){\rm d}\tau=\infty$ holds, in particular, if $T<\infty$ and $\rho(t)\le L(T-t)$ for some $L>0$ \emph{(a clock vanishing in finite time)}, or if $T=\infty$ and $\rho(t)\le L(1+t)$ for some $L>0$ \emph{(a clock diverging at most linearly)};
  \item[(ii)] \emph{(Disturbance budget in finite-time case)} If $T<\infty$, then
  \begin{equation}\label{eq:dil-mass}
    \int_0^\infty \rho(t(s))\,{\rm d}s=T.
  \end{equation}
  Consequently, for every signal $b\in L^\infty([0,\infty))$, $\rho b\in L^1([0,\infty))$. Moreover, if $w$ satisfies the scalar filter ${\rm d}w/{\rm d}s=-a(s)w+\rho b$, where $a(s)\ge\lambda>0$, $b\in L^\infty([0,\infty))$, and $w(0)\in\mathbb R$, then $w\in L^1([0,\infty))$;
  \item[(iii)] \emph{(Convergence)} One has $\delta(t)\to 0$ as $t\to T^-$ in either of the following cases:
  \begin{enumerate}
    \item[(a)] \emph{(Integrable input)} $d\in L^1([0,\infty))$, and the regular system \eqref{eq:dil-desing} admits an iISS Lyapunov estimate in the sense of Lemma~\ref{lem:iiss} (Appendix~\ref{app:iiss}) with linearly bounded supply. Namely, there exist a $\mathcal C^1$ proper positive-definite function $V$, a positive-definite function $\alpha$, and $C>0$ such that  ${\rm d}V/{\rm d}s\le-\alpha(|\delta|)+C|d|$ along \eqref{eq:dil-desing}. In this case,  in addition, there exist $\beta\in\mathcal{KL}$ and $\chi\in\mathcal{K_\infty}$ such that, for all $s\ge s_0\ge0$,
    \begin{equation}\label{eq:dil-est}
      |\delta(s)|\le\beta(|\delta(s_0)|,s-s_0)+\chi\left(\int_{s_0}^s|d(\tau)|{\rm d}\tau\right).
    \end{equation}

    \item[(b)] \emph{(Vanishing input)} $d(s)\to0$ as $s\to\infty$, and $f$ satisfies $xf(x)<0$ for all $x\ne0$, and $\liminf_{|x|\to\infty}|f(x)|>0$.
  \end{enumerate}
\end{enumerate}
\end{lemma}

\begin{proof}
By a slight abuse of notation, we denote $\rho(s):=\rho(t(s))$ and $\delta(s):=\delta(t(s))$.

\emph{(i)} Since $\rho(t)>0$, the map $s$ is strictly increasing, hence invertible with inverse $t(s)$. Moreover, by assumption, $s(t)\to\int_0^T\rho^{-1}(\tau){\rm d}\tau=\infty$ as $t\to T^-$. Thus, $s$ maps $[0,T)$ bijectively onto $[0,\infty)$.
Differentiating $\delta$ in $s$ and using ${\rm d}t/{\rm d}s=\rho(t(s))$ gives ${{\rm d}\delta}/{{\rm d}s}=\rho(t(s))\dot\delta(t(s))=f(\delta)+p(t(s))$, which is \eqref{eq:dil-desing}: the factor $1/\rho$ of \eqref{eq:dil-sys} is absorbed and no singularity remains. 
The two sufficient conditions imply divergence of the clock integral. Indeed, if $\rho(t)\le L(T-t)$ then $\int_0^t\rho^{-1}(\tau){\rm d}\tau\ge\frac1L\int_0^t\frac{{\rm d}\tau}{T-\tau}=\frac1L\ln\frac{T}{T-t}\to\infty$, whereas if $\rho(t)\le L(1+t)$ then $\int_0^t\rho^{-1}(\tau){\rm d}\tau\ge\frac1L\int_0^t\frac{{\rm d}\tau}{1+\tau}=\frac1L\ln(1+t)\to\infty$.

\emph{(ii)} For $T<\infty$, the change of variable ${\rm d}s={\rm d}t/\rho$ gives $\int_0^\infty\rho(t(s))\,{\rm d}s=\int_0^T{\rm d}t=T$, which is \eqref{eq:dil-mass}. 
Hence, for every $b\in L^\infty([0,\infty))$, $\int_0^\infty|\rho(s)b(s)|{\rm d}s\le\|b\|_\infty T<\infty$. Thus, $\rho b\in L^1([0,\infty))$.

For the filter output $w$, the transition factor obeys $\exp(-\int_l^s a(\tau){\rm d}\tau)\le e^{-\lambda(s-l)}$, so the comparison principle gives $|w(s)|\le e^{-\lambda s}|w_0|+\int_0^s e^{-\lambda(s-l)}\rho(l)|b(l)|\,{\rm d}l$. Integrating and exchanging the order of integration yield
\begin{equation*}
  \int_0^\infty|w(s)|\,{\rm d}s
  \le\frac{|w_0|}{\lambda}+\int_0^\infty\rho(l)|b(l)|\left(\int_l^\infty e^{-\lambda(s-l)}\,{\rm d}s\right){\rm d}l
  =\frac{|w_0|}{\lambda}+\frac1\lambda\int_0^\infty\rho(l)|b(l)|\,{\rm d}l\le\frac{|w_0|}{\lambda}+\frac{\|b\|_\infty T}{\lambda}<\infty .
\end{equation*}
Thus, $w\in L^1([0,\infty))$. 

\emph{(iii)-(a)} The hypothesis supplies an iISS Lyapunov function for the regular system \eqref{eq:dil-desing} with linearly bounded supply rate. Lemma~\ref{lem:iiss} (Appendix~\ref{app:iiss}) then furnishes $\beta\in\mathcal{KL}$ and $\chi\in\mathcal K_\infty$ for which \eqref{eq:dil-est} holds for all $s\ge s_0\ge0$. The supply rate being linearly bounded, $d\in L^1([0,\infty))$ yields $\int_0^\infty|d(s)|{\rm d}s<\infty$. Taking $s_0=0$, the $\mathcal{KL}$ term is at most $\beta(|\delta_0|,0)$ and the gain term at most $\chi\left(\int_0^\infty|d(s)|{\rm d}s\right)$, both finite, so $\delta$ is bounded. Fixing $s_0\ge0$ and letting $s\to\infty$, the $\mathcal{KL}$ term vanishes, leaving $\limsup_{s\to\infty}|\delta(s)|\le\chi\left(\int_{s_0}^\infty|d(s)|{\rm d}s\right)$. The left side is independent of $s_0$, so letting $s_0\to\infty$ and using that the tail of the convergent integral vanishes (with $\chi(0)=0$) gives $\delta(s)\to0$.

\emph{(iii)-(b)} Fix any $l>0$. By continuity of $f$, $f(x)\ne0$ for $x\ne0$, and $\liminf_{|x|\to\infty}|f(x)|>0$, the quantity $m_l:=\inf_{|x|\ge l}|f(x)|$ is positive. Since $d(s)\to0$, there is $S_l$ with $|d(s)|\le m_l/2$ for $s\ge S_l$. For $\delta\ge l$, the sign condition $xf(x)<0$ and $|f|\ge m_l$ give $f(\delta)\le-m_l$, so ${{\rm d}\delta}/{{\rm d}s}=f(\delta)+d(s)\le-m_l/2$ for $s\ge S_l$. Symmetrically, ${{\rm d}\delta}/{{\rm d}s}\ge m_l/2$ when $\delta\le-l$. Hence, for $s\ge S_l$, the right-hand side points strictly into $[-l,l]$ whenever $|\delta|\ge l$, so $\delta$ enters $[-l,l]$ in finite $s$ and remains there: $\limsup_{s\to\infty}|\delta(s)|\le l$. As $l>0$ is arbitrary, $\delta(s)\to0$.

In either case, $t(s)\to T^-$ as $s\to\infty$, hence $\delta(t)\to0$ as $t\to T^-$. \qed
\end{proof}

\appsection{Persistence of a margin under a singular gain}\label{app:safety}

This appendix proves the safety counterpart of Lemma \ref{lem:capture}. Here, the singular gain $c/r$ enforces persistence ($r$ never reaches zero) rather than finite-time collapse. The coupling, in which the same singular gain also drives the regulated error $e$, is resolved by passing to the margin-generated time $\sigma=\int_0^t r^{-1}$: the ratio $\zeta=e/r$ contracts, so $|e|\lesssim r$, and the perturbation $e$ exerts on $\dot r$ vanishes as $r\to 0^+$.

\begin{lemma}[Persistence of a margin under a singular gain]\label{lem:safety}
Let $T\in(0,\infty]$, and let $r:[0,T)\to(0,\infty)$ and $e:[0,T)\to\mathbb{R}$ be absolutely continuous. Suppose there exist constants $k,c,L,\mu,\kappa>0$ and $B\ge0$ with $c>L$, such that, for almost every $t\in[0,T)$,
\begin{eqnarray}
  \frac{\rm d}{{\rm d}t}|e| &\le& -\left(k+\frac{c}{r}\right)|e| + B, \label{eq:saf-e}\\
  \dot r &\ge& \max\left\{-L, \mu-\kappa\left(|e|+\frac12 e^2\right)\right\}. \label{eq:saf-r}
\end{eqnarray}
Define $\nu:=c-L$, $r_\star:={e^\star}/{\zeta_M}$ (with the convention $r_\star:=\infty$ if $\zeta_M=0$), $\underline r:=\min\{r(0),r_\star\}$,
\begin{equation}\label{eq:saf-data}
  e^\star:=\sqrt{1+\frac{\mu}{\kappa}}-1,\quad
  \zeta_M:=\max\left\{\frac{|e(0)|}{r(0)},\frac{B}{\nu}\right\},\quad
  T_\eta:=\max\left\{0, \frac{1}{k}\left(\ln\frac{|e(0)|}{e^\star-B/k}\right) \right\}
\end{equation}
with the convention $\ln(0):=-\infty$. Assume the gain condition $B/k<e^\star$. Then, for all $t\in[0,T)$:
\begin{itemize}
  \item[(i)] $|e(t)|\le\max\{|e(0)|,B/k\}$, and $|e(t)|\le e^\star$ for $t\in[T_\eta,T)$;
  \item[(ii)] $|e(t)|\le\zeta_M\,r(t)$;
  \item[(iii)] $r(t)\ge\underline r>0$, and $r(t)\ge\underline r+{\mu}(t-T_\eta)/2$ for $t\in[T_\eta,T)$.
\end{itemize}
\end{lemma}

\noindent
\textit{Interpretation.}
Read $r$ as the distance to a level the system must never reach and $e$ as a steering error. The gain $c/r$ in the error dynamics grows without bound as $r\to0^+$. Far from zero, it merely damps $e$; near zero, it forces $e$ to shrink in proportion to $r$ itself, so the error can no longer push $r$ down. The nominal outward drift $\mu$ then prevails, and $r$ is held away from zero for all time, escaping at least linearly once the error has settled.

\begin{proof}
\emph{Step 1 (Inner ultimate bound).}
On $[0,T)$ we have $c/r>0$, so \eqref{eq:saf-e} gives $\frac{\rm d}{{\rm d}t}|e|\le-k|e|+B$. By the comparison principle $|e(t)|\le\bar{e}(t):=B/k+e^{-kt}\left(|e(0)|-B/k\right)\le\max\{|e(0)|,B/k\}$. 
For $t\ge T_\eta$, one gets $e^{-kt}|e(0)|\le e^\star-B/k$. Since $B/k<e^\star$, the majorant $\bar{e}$ satisfies $\bar{e}(t)\le e^\star$ for $t\ge T_\eta$, giving~(i).

\emph{Step 2 (Bounded ratio).}
Introduce the margin-generated time $\sigma(t):=\int_0^t r(\tau)^{-1}\,{\rm d}\tau$ and the ratio $\zeta:=e/r$ on $[0,T)$. 
Since ${\rm d}\sigma=r^{-1}{\rm d}t$
and $|\zeta|=|e|/r$, one has $\frac{{\rm d}}{{\rm d}\sigma}|\zeta|=\frac{{\rm d}|e|}{{\rm d}t}-|\zeta|\,\dot r$. Using \eqref{eq:saf-e} and the identity $\left(k+{c}/{r}\right)|e|=(kr+c)|\zeta|$,
\begin{equation*}
  \frac{{\rm d}}{{\rm d}\sigma}|\zeta|
  =\frac{{\rm d}|e|}{{\rm d}t}-|\zeta|\,\dot r
  \le -(c+kr)|\zeta|+B-|\zeta|\,\dot r
  = -\left(c+kr+\dot r\right)|\zeta|+B .
\end{equation*}
Since $kr>0$ and $\dot r\ge-L$ by \eqref{eq:saf-r}, the bracket satisfies $c+kr+\dot r\ge c-L=\nu>0$, whence $\frac{{\rm d}}{{\rm d}\sigma}|\zeta|\le-\nu|\zeta|+B$. The comparison principle gives $|\zeta(\sigma)|\le e^{-\nu\sigma}|\zeta(0)|+{B}(1-e^{-\nu\sigma})/\nu\le\max\{|\zeta(0)|,B/\nu\}=\zeta_M$, the second inequality because the middle expression is a convex combination of $|\zeta(0)|$ and $B/\nu$. Hence $|e(t)|\le\zeta_M\,r(t)$, which is~(ii).

\emph{Step 3 (Persistence and escape).}
Substituting~(ii) into \eqref{eq:saf-r} yields
\begin{equation*}
  \dot r\ge\mu-\kappa\left(|e|+\frac12 e^2\right)\ge\mu-\kappa\left(\zeta_M r+\frac12\zeta_M^2 r^2\right).
\end{equation*}
The right-hand side is at least $\mu/2$ precisely when $\zeta_M r+\frac12\zeta_M^2 r^2\le\mu/(2\kappa)$, i.e., when $\zeta_M r\le\sqrt{1+\mu/\kappa}-1$, that is, for $0<r\le r_\star$ (with the convention $r_\star=\infty$ when $\zeta_M=0$ understood here). Hence $\dot r\ge\mu/2>0$ whenever $0<r\le r_\star$, so $r$ cannot decrease through $\underline r=\min\{r(0),r_\star\}$, proving $r(t)\ge\underline r>0$ for all $t\in[0,T)$. Finally, by~(i), for $t\in[T_\eta,T)$, we have $|e|\le e^\star$, so $\kappa(|e|+ e^2/2)\le\kappa(e^\star+(e^\star)^2/2)=\mu/2$ and \eqref{eq:saf-r} gives $\dot r\ge\mu/2$; integrating from $T_\eta$ yields $r(t)\ge r(T_\eta)+{\mu}(t-T_\eta)/2$, and combining with the persistence bound $r(T_\eta)\ge\underline r$ already proved gives $r(t)\ge\underline r+{\mu}(t-T_\eta)/2$, completing~(iii). \qed
\end{proof}

\appsection{An iISS convergence lemma under an integrable input}\label{app:iiss}

\begin{lemma}[Angeli--Sontag--Wang \cite{angeli2002characterization}]\label{lem:iiss}
Consider $\dot x=F(x,d(s))$ on $s\ge 0$, with $F$ continuous, locally Lipschitz in $x$, $F(0,0)=0$, and $d$ measurable and locally essentially bounded. If there exist a $\mathcal C^1$ proper positive-definite function $V$, a positive-definite function $\alpha$, and $\sigma\in\mathcal K$ such that $\dot V\le-\alpha(|x|)+\sigma(|d(s)|)$ along solutions, then the system is iISS: there exist $\beta\in\mathcal{KL}$ and $\chi_1,\chi_2\in\mathcal K_\infty$ with $|x(s)|\le\beta(|x(s_0)|,s-s_0)+\chi_1\!\left(\int_{s_0}^{s}\chi_2(|d(\tau)|)\,{\rm d}\tau\right)$ for all $s\ge s_0\ge 0$; in particular, if $\int_0^\infty\chi_2(|d(\tau)|)\,{\rm d}\tau<\infty$ (e.g., $d\in L^1$ and $\sigma$ linearly bounded), then $x(s)\to 0$.
\end{lemma}

\bibliographystyle{model1-num-names}

\bibliography{v1-arXiv}

\begin{thebibliography}{34}
\expandafter\ifx\csname natexlab\endcsname\relax\def\natexlab#1{#1}\fi
\providecommand{\bibinfo}[2]{#2}
\ifx\xfnm\relax \def\xfnm[#1]{\unskip,\space#1}\fi
\bibitem[{Krstić(2026)}]{krstic2026directional}
\bibinfo{author}{M.~Krstić},
\newblock \bibinfo{title}{Directional pursuit and evasion: Input-to-state stabilizing and inverse-minimax feedbacks},
\newblock \bibinfo{journal}{Journal of Guidance, Control, and Dynamics}  (\bibinfo{year}{2026}) \bibinfo{pages}{1--12}.
\bibitem[{Isaacs(1965)}]{isaacs1965differential}
\bibinfo{author}{R.~Isaacs}, \bibinfo{title}{Differential Games: A Mathematical Theory with Applications to Warfare and Pursuit, Control and Optimization}, \bibinfo{publisher}{Wiley}, \bibinfo{address}{London, U.K.}, \bibinfo{year}{1965}.
\bibitem[{Ba{\c{s}}ar and Olsder(1998)}]{bacsar1998dynamic}
\bibinfo{author}{T.~Ba{\c{s}}ar}, \bibinfo{author}{G.~J. Olsder}, \bibinfo{title}{Dynamic noncooperative game theory}, \bibinfo{publisher}{SIAM}, \bibinfo{address}{Philadelphia}, \bibinfo{year}{1998}.
\bibitem[{Exarchos et~al.(2015)Exarchos, Tsiotras, and Pachter}]{exarchos2015suicidal}
\bibinfo{author}{I.~Exarchos}, \bibinfo{author}{P.~Tsiotras}, \bibinfo{author}{M.~Pachter},
\newblock \bibinfo{title}{On the suicidal pedestrian differential game},
\newblock \bibinfo{journal}{Dynamic Games and Applications} \bibinfo{volume}{5} (\bibinfo{year}{2015}) \bibinfo{pages}{297--317}.
\bibitem[{Merz(1972)}]{merz1972game}
\bibinfo{author}{A.~Merz},
\newblock \bibinfo{title}{The game of two identical cars},
\newblock \bibinfo{journal}{Journal of Optimization Theory and Applications} \bibinfo{volume}{9} (\bibinfo{year}{1972}) \bibinfo{pages}{324--343}.
\bibitem[{Buzikov and Galyaev(2023)}]{buzikov2023game}
\bibinfo{author}{M.~Buzikov}, \bibinfo{author}{A.~Galyaev},
\newblock \bibinfo{title}{The game of two identical cars: An analytical description of the barrier},
\newblock \bibinfo{journal}{Journal of Optimization Theory and Applications} \bibinfo{volume}{198} (\bibinfo{year}{2023}) \bibinfo{pages}{988--1018}.
\bibitem[{Bera et~al.(2017)Bera, Makkapati, and Kothari}]{bera2017comprehensive}
\bibinfo{author}{R.~Bera}, \bibinfo{author}{V.~R. Makkapati}, \bibinfo{author}{M.~Kothari},
\newblock \bibinfo{title}{A comprehensive differential game theoretic solution to a game of two cars},
\newblock \bibinfo{journal}{Journal of Optimization Theory and Applications} \bibinfo{volume}{174} (\bibinfo{year}{2017}) \bibinfo{pages}{818--836}.
\bibitem[{Mitchell et~al.(2005)Mitchell, Bayen, and Tomlin}]{mitchell2005time}
\bibinfo{author}{I.~M. Mitchell}, \bibinfo{author}{A.~M. Bayen}, \bibinfo{author}{C.~J. Tomlin},
\newblock \bibinfo{title}{A time-dependent {H}amilton-{J}acobi formulation of reachable sets for continuous dynamic games},
\newblock \bibinfo{journal}{IEEE Trans. Autom. Contr.} \bibinfo{volume}{50} (\bibinfo{year}{2005}) \bibinfo{pages}{947--957}.
\bibitem[{Sun et~al.(2017)Sun, Tsiotras, Lolla, Subramani, and Lermusiaux}]{sun2017multiple}
\bibinfo{author}{W.~Sun}, \bibinfo{author}{P.~Tsiotras}, \bibinfo{author}{T.~Lolla}, \bibinfo{author}{D.~N. Subramani}, \bibinfo{author}{P.~F. Lermusiaux},
\newblock \bibinfo{title}{Multiple-pursuer/one-evader pursuit--evasion game in dynamic flowfields},
\newblock \bibinfo{journal}{Journal of guidance, control, and dynamics} \bibinfo{volume}{40} (\bibinfo{year}{2017}) \bibinfo{pages}{1627--1637}.
\bibitem[{Chung et~al.(2011)Chung, Hollinger, and Isler}]{chung2011search}
\bibinfo{author}{T.~H. Chung}, \bibinfo{author}{G.~A. Hollinger}, \bibinfo{author}{V.~Isler},
\newblock \bibinfo{title}{Search and pursuit-evasion in mobile robotics: A survey},
\newblock \bibinfo{journal}{Autonomous robots} \bibinfo{volume}{31} (\bibinfo{year}{2011}) \bibinfo{pages}{299--316}.
\bibitem[{Guibas et~al.(1999)Guibas, Latombe, LaValle, Lin, and Motwani}]{guibas1999visibility}
\bibinfo{author}{L.~J. Guibas}, \bibinfo{author}{J.-C. Latombe}, \bibinfo{author}{S.~M. LaValle}, \bibinfo{author}{D.~Lin}, \bibinfo{author}{R.~Motwani},
\newblock \bibinfo{title}{A visibility-based pursuit-evasion problem},
\newblock \bibinfo{journal}{International Journal of Computational Geometry \& Applications} \bibinfo{volume}{9} (\bibinfo{year}{1999}) \bibinfo{pages}{471--493}.
\bibitem[{Gerkey et~al.(2006)Gerkey, Thrun, and Gordon}]{gerkey2006visibility}
\bibinfo{author}{B.~P. Gerkey}, \bibinfo{author}{S.~Thrun}, \bibinfo{author}{G.~Gordon},
\newblock \bibinfo{title}{Visibility-based pursuit-evasion with limited field of view},
\newblock \bibinfo{journal}{The International Journal of Robotics Research} \bibinfo{volume}{25} (\bibinfo{year}{2006}) \bibinfo{pages}{299--315}.
\bibitem[{Lozano et~al.(2022)Lozano, Becerra, Ruiz, Bravo, and Murrieta-Cid}]{lozano2022visibility}
\bibinfo{author}{E.~Lozano}, \bibinfo{author}{I.~Becerra}, \bibinfo{author}{U.~Ruiz}, \bibinfo{author}{L.~Bravo}, \bibinfo{author}{R.~Murrieta-Cid},
\newblock \bibinfo{title}{A visibility-based pursuit-evasion game between two nonholonomic robots in environments with obstacles},
\newblock \bibinfo{journal}{Autonomous Robots} \bibinfo{volume}{46} (\bibinfo{year}{2022}) \bibinfo{pages}{349--371}.
\bibitem[{Ruiz et~al.(2013)Ruiz, Murrieta-Cid, and Marroquin}]{ruiz2013time}
\bibinfo{author}{U.~Ruiz}, \bibinfo{author}{R.~Murrieta-Cid}, \bibinfo{author}{J.~L. Marroquin},
\newblock \bibinfo{title}{Time-optimal motion strategies for capturing an omnidirectional evader using a differential drive robot},
\newblock \bibinfo{journal}{IEEE Trans. Robot.} \bibinfo{volume}{29} (\bibinfo{year}{2013}) \bibinfo{pages}{1180--1196}.
\bibitem[{Ruiz and Isler(2016)}]{ruiz2016capturing}
\bibinfo{author}{U.~Ruiz}, \bibinfo{author}{V.~Isler},
\newblock \bibinfo{title}{Capturing an omnidirectional evader in convex environments using a differential drive robot},
\newblock \bibinfo{journal}{IEEE Robot. Autom. Lett.} \bibinfo{volume}{1} (\bibinfo{year}{2016}) \bibinfo{pages}{1007--1013}.
\bibitem[{Bravo et~al.(2020)Bravo, Ruiz, and Murrieta-Cid}]{bravo2020pursuit}
\bibinfo{author}{L.~Bravo}, \bibinfo{author}{U.~Ruiz}, \bibinfo{author}{R.~Murrieta-Cid},
\newblock \bibinfo{title}{A pursuit--evasion game between two identical differential drive robots},
\newblock \bibinfo{journal}{J. Franklin Institute} \bibinfo{volume}{357} (\bibinfo{year}{2020}) \bibinfo{pages}{5773--5808}.
\bibitem[{Dovrat et~al.(2021)Dovrat, Tripathy, and Bruckstein}]{dovrat2021tracking}
\bibinfo{author}{D.~Dovrat}, \bibinfo{author}{T.~Tripathy}, \bibinfo{author}{A.~M. Bruckstein},
\newblock \bibinfo{title}{On tracking and capture in proportional-control bearing-only unicycle pursuit},
\newblock \bibinfo{journal}{IEEE Contr. Syst. Lett.} \bibinfo{volume}{6} (\bibinfo{year}{2021}) \bibinfo{pages}{2132--2137}.
\bibitem[{Nath and Ghose(2022)}]{nath2022two}
\bibinfo{author}{S.~Nath}, \bibinfo{author}{D.~Ghose},
\newblock \bibinfo{title}{A two-phase evasive strategy for a pursuit-evasion problem involving two non-holonomic agents with incomplete information},
\newblock \bibinfo{journal}{Eur. J. Control} \bibinfo{volume}{68} (\bibinfo{year}{2022}) \bibinfo{pages}{100677}.
\bibitem[{Sani et~al.(2020)Sani, Robu, and Hably}]{sani2020pursuit}
\bibinfo{author}{M.~Sani}, \bibinfo{author}{B.~Robu}, \bibinfo{author}{A.~Hably},
\newblock \bibinfo{title}{Pursuit-evasion game for nonholonomic mobile robots with obstacle avoidance using {NMPC}},
\newblock in: \bibinfo{booktitle}{2020 28th Mediterranean conference on control and automation (MED)}, \bibinfo{organization}{IEEE}, pp. \bibinfo{pages}{978--983}.
\bibitem[{Zhou et~al.(2026)Zhou, Li, Zhao, Wahlberg, and Hu}]{zhou2026nature}
\bibinfo{author}{P.~Zhou}, \bibinfo{author}{S.~Li}, \bibinfo{author}{B.~Zhao}, \bibinfo{author}{B.~Wahlberg}, \bibinfo{author}{X.~Hu},
\newblock \bibinfo{title}{Nature-inspired dynamic control for pursuit-evasion of robots},
\newblock \bibinfo{journal}{Automatica} \bibinfo{volume}{183} (\bibinfo{year}{2026}) \bibinfo{pages}{112629}.
\bibitem[{Kreindler(1973)}]{kreindler1973optimality}
\bibinfo{author}{E.~Kreindler},
\newblock \bibinfo{title}{Optimality of proportional navigation},
\newblock \bibinfo{journal}{AIAA Journal} \bibinfo{volume}{11} (\bibinfo{year}{1973}) \bibinfo{pages}{878--880}.
\bibitem[{Ryoo et~al.(2004)Ryoo, Kim, and Tahk}]{ryoo2004capturability}
\bibinfo{author}{C.-K. Ryoo}, \bibinfo{author}{Y.-H. Kim}, \bibinfo{author}{M.-J. Tahk},
\newblock \bibinfo{title}{Capturability analysis of pn laws using {L}yapunov stability theory},
\newblock in: \bibinfo{booktitle}{AIAA Guidance, Navigation, and Control Conference and Exhibit}, p. \bibinfo{pages}{4883}.
\bibitem[{Palumbo et~al.(2010)Palumbo, Blauwkamp, and Lloyd}]{palumbo2010modern}
\bibinfo{author}{N.~F. Palumbo}, \bibinfo{author}{R.~A. Blauwkamp}, \bibinfo{author}{J.~M. Lloyd},
\newblock \bibinfo{title}{Modern homing missile guidance theory and techniques},
\newblock \bibinfo{journal}{Johns Hopkins APL technical digest} \bibinfo{volume}{29} (\bibinfo{year}{2010}) \bibinfo{pages}{42--59}.
\bibitem[{Krsti{\'c} and Li(1998)}]{krstic1998inverse}
\bibinfo{author}{M.~Krsti{\'c}}, \bibinfo{author}{Z.-H. Li},
\newblock \bibinfo{title}{Inverse optimal design of input-to-state stabilizing nonlinear controllers},
\newblock \bibinfo{journal}{IEEE Trans. Autom. Contr.} \bibinfo{volume}{43} (\bibinfo{year}{1998}) \bibinfo{pages}{336--350}.
\bibitem[{Han and Wang(2024)}]{han2024safety}
\bibinfo{author}{T.~Han}, \bibinfo{author}{B.~Wang},
\newblock \bibinfo{title}{Safety-critical stabilization of force-controlled nonholonomic mobile robots},
\newblock \bibinfo{journal}{IEEE Contr. Syst. Lett.} \bibinfo{volume}{8} (\bibinfo{year}{2024}) \bibinfo{pages}{2469--2474}.
\bibitem[{Wang et~al.(2026)Wang, Han, and Wang}]{wang2026further}
\bibinfo{author}{B.~Wang}, \bibinfo{author}{T.~Han}, \bibinfo{author}{G.~Wang},
\newblock \bibinfo{title}{Further results on safety-critical stabilization of force-controlled nonholonomic mobile robots},
\newblock \bibinfo{journal}{ASME Letters in Dynamic Systems and Control} \bibinfo{volume}{6} (\bibinfo{year}{2026}) \bibinfo{pages}{021011}.
\bibitem[{Todorovski et~al.(2025)Todorovski, Kim, and Krsti{\'c}}]{todorovski2025modular}
\bibinfo{author}{V.~Todorovski}, \bibinfo{author}{K.~H. Kim}, \bibinfo{author}{M.~Krsti{\'c}},
\newblock \bibinfo{title}{Modular design of strict control {L}yapunov functions for global stabilization of the unicycle in polar coordinates},
\newblock \bibinfo{journal}{arXiv preprint arXiv:2509.25575}  (\bibinfo{year}{2025}).
\bibitem[{Todorovski et~al.(2026)Todorovski, Kim, Astolfi, and Krstić}]{todorovski2026nonholonomicrobotparkingfeedback}
\bibinfo{author}{V.~Todorovski}, \bibinfo{author}{K.~H. Kim}, \bibinfo{author}{A.~Astolfi}, \bibinfo{author}{M.~Krstić}, \bibinfo{title}{Nonholonomic robot parking by feedback---{P}art {I}: Modular strict {CLF} designs}, \bibinfo{year}{2026}.
\bibitem[{Wang et~al.(2021)Wang, Nersesov, and Ashrafiuon}]{wang2021formation}
\bibinfo{author}{B.~Wang}, \bibinfo{author}{S.~Nersesov}, \bibinfo{author}{H.~Ashrafiuon},
\newblock \bibinfo{title}{Formation regulation and tracking control for nonholonomic mobile robot networks using polar coordinates},
\newblock \bibinfo{journal}{IEEE Contr. Syst. Lett.} \bibinfo{volume}{6} (\bibinfo{year}{2021}) \bibinfo{pages}{1909--1914}.
\bibitem[{Krsti{\'c} et~al.(2025)Krsti{\'c}, Todorovski, Kim, and Astolfi}]{krstic2025integrator}
\bibinfo{author}{M.~Krsti{\'c}}, \bibinfo{author}{V.~Todorovski}, \bibinfo{author}{K.~H. Kim}, \bibinfo{author}{A.~Astolfi},
\newblock \bibinfo{title}{Integrator forwarding design for unicycles with constant and actuated velocity in polar coordinates},
\newblock \bibinfo{journal}{arXiv preprint arXiv:2509.25579}  (\bibinfo{year}{2025}).
\bibitem[{Kim et~al.(2025{\natexlab{a}})Kim, Todorovski, and Krsti{\'c}}]{kim2025inverse}
\bibinfo{author}{K.~H. Kim}, \bibinfo{author}{V.~Todorovski}, \bibinfo{author}{M.~Krsti{\'c}},
\newblock \bibinfo{title}{Inverse optimal feedback and gain margins for unicycle stabilization},
\newblock \bibinfo{journal}{arXiv preprint arXiv:2509.25563}  (\bibinfo{year}{2025}{\natexlab{a}}).
\bibitem[{Kim et~al.(2025{\natexlab{b}})Kim, Todorovski, and Krstić}]{kim2025nonholonomicrobotparkingfeedback}
\bibinfo{author}{K.~H. Kim}, \bibinfo{author}{V.~Todorovski}, \bibinfo{author}{M.~Krstić}, \bibinfo{title}{Nonholonomic robot parking by feedback---{P}art {II}: Nonmodular, inverse optimal, adaptive, prescribed/fixed-time and safe designs}, \bibinfo{year}{2025}{\natexlab{b}}.
\bibitem[{Karafyllis and Krstić(2025)}]{karafyllis2025robust}
\bibinfo{author}{I.~Karafyllis}, \bibinfo{author}{M.~Krstić}, \bibinfo{title}{Robust Adaptive Control: Deadzone-Adapted Disturbance Suppression}, \bibinfo{publisher}{SIAM}, \bibinfo{address}{Philadelphia, PA}, \bibinfo{year}{2025}.
\bibitem[{Angeli et~al.(2000)Angeli, Sontag, and Wang}]{angeli2002characterization}
\bibinfo{author}{D.~Angeli}, \bibinfo{author}{E.~D. Sontag}, \bibinfo{author}{Y.~Wang},
\newblock \bibinfo{title}{A characterization of integral input-to-state stability},
\newblock \bibinfo{journal}{IEEE Trans. Autom. Contr.} \bibinfo{volume}{45} (\bibinfo{year}{2000}) \bibinfo{pages}{1082--1097}.

\end{thebibliography}

\end{document}